\documentclass[11pt]{article}

\usepackage[a4paper,left=20mm,top=20mm,right=20mm,bottom=25mm]{geometry}

\usepackage{amsmath, amsfonts, amssymb, amsthm}
\usepackage{algorithm}
\usepackage[noend]{algorithmic}
\usepackage{algorithm}
\usepackage{authblk}

\usepackage{graphicx}  
\usepackage{mathtools}
\usepackage{subfig}
\usepackage[mathscr]{euscript}
\usepackage{enumitem}
\usepackage[normalem]{ulem}

\usepackage[colorlinks]{hyperref}
\hypersetup{
  citecolor=blue,
  linkcolor=blue,
}

\usepackage[font=footnotesize,width=.85\textwidth,labelfont=bf]{caption}

\usepackage{color}

\usepackage{wasysym}

\usepackage{authblk}
\usepackage[numbers, sort&compress]{natbib}

\newtheorem{theorem}{Theorem}[section]
\newtheorem{lemma}[theorem]{Lemma}
\newtheorem{corollary}[theorem]{Corollary}
\newtheorem{definition}[theorem]{Definition}

\newtheorem{proposition}[theorem]{Proposition}
\newtheorem{fact}[theorem]{Observation}
\newtheorem{remark}[theorem]{Remark}
\newtheorem{claim}{Claim}[theorem]
\newtheorem{problem}{Problem}[section]
\newcommand{\PROBLEM}[1]{{\sc #1}}

\newenvironment{claim-proof}%
    {\begin{description}[leftmargin = 0.2cm, labelsep = 0.2cm]}
    {\end{description}}

\makeatletter
\def\moverlay{\mathpalette\mov@rlay}
\def\mov@rlay#1#2{\leavevmode\vtop{%
   \baselineskip\z@skip \lineskiplimit-\maxdimen
   \ialign{\hfil$\m@th#1##$\hfil\cr#2\crcr}}}
\newcommand{\charfusion}[3][\mathord]{
    #1{\ifx#1\mathop\vphantom{#2}\fi
        \mathpalette\mov@rlay{#2\cr#3}
      }
    \ifx#1\mathop\expandafter\displaylimits\fi}
\makeatother

\DeclareMathOperator{\lca}{lca}

\DeclareMathOperator{\symdiff}{\triangle}
\newcommand{\cupdot}{\charfusion[\mathbin]{\cup}{\cdot}}
\newcommand{\hourglass}{\mathrel{\text{\ooalign{$\searrow$\cr$\nearrow$}}}}

\usepackage{xcolor}
\definecolor{light-gray}{gray}{0.6}
\definecolor{violet}{RGB}{160,0,160}
\definecolor{olive}{RGB}{128,128,0}

\newcommand{\undirected}{\widetilde}
\newcommand{\Black}{\mathrm{black}}
\newcommand{\White}{\mathrm{white}}

\newcommand{\AX}[1]{\textnormal{#1}}

\newcommand{\G}{G}

\newcommand{\SPEC}{\newmoon}

\newcommand{\DUPL}{\square}
\DeclareMathOperator{\Aho}{Aho}
\newcommand{\child}{\mathsf{child}}

\newcommand{\rthin}{\mathrel{\mathrel{\ooalign{\hss\raisebox{-0.17ex}{$\sim$}\hss\cr\hss\raisebox{0.720ex}{\scalebox{0.75}{$\bullet$}}\hss}}}}

\providecommand{\keywords}[1]{\textbf{\textit{Keywords: }} #1}

\title{Complexity of modification problems  for best match graphs}

\author[1,2]{David Schaller}
\author[1-5]{Peter F.\ Stadler}
\author[6,7]{Marc Hellmuth}

\affil[1]{Max Planck Institute for Mathematics in the Sciences,
  Inselstra{\ss}e 22, D-04103 Leipzig, Germany}

\affil[2]{Bioinformatics Group, Department of Computer Science \&
  Interdisciplinary Center for Bioinformatics, Universit{\"a}t Leipzig,
  H{\"a}rtelstra{\ss}e~16--18, D-04107 Leipzig, Germany.}

\affil[3]{Institute for Theoretical Chemistry, University of Vienna,
  W{\"a}hringerstrasse 17, A-1090 Wien, Austria}

\affil[4]{Facultad de Ciencias, Universidad National de Colombia, Sede
  Bogot{\'a}, Colombia}

\affil[5]{Santa Fe Insitute, 1399 Hyde Park Rd., Santa Fe NM 87501,
  USA}

\affil[6]{School of Computing, University of Leeds, EC Stoner
  Building, Leeds LS2 9JT, UK}

\affil[7]{Department of Mathematics, Faculty of Science,
  Stockholm University, SE - 106 91 Stockholm,   Sweden \newline 
  \texttt{mhellmuth@mailbox.org}}

\date{\ }

\begin{document}

\maketitle

\abstract{ 
  Best match graphs (BMGs) are vertex-colored directed graphs that were
  introduced to model the relationships of genes (vertices) from different
  species (colors) given an underlying evolutionary tree that is assumed to
  be unknown. In real-life applications, BMGs are estimated from sequence
  similarity data. Measurement noise and approximation errors usually
  result in empirically determined graphs that in general violate
  characteristic properties of BMGs. The arc modification problems for BMGs 
	aim at correcting such violations and thus provide a means to 
  improve the initial estimates of best match data. We show here that
  the arc deletion, arc completion and arc editing problems for BMGs are
  NP-complete and that they can be formulated and solved as integer linear
  programs.  To this end, we provide a novel characterization of BMGs
  in terms of triples (binary trees on three leaves) and a
  characterization of BMGs with two colors in terms of forbidden subgraphs.
}

\bigskip
\noindent

\sloppy

\section{Introduction}

Best match graphs (BMGs) appear in mathematical biology as formal
description of the evolutionary relationships within a gene family. Each
vertex $x$ represents a gene and is ``colored'' by the species $\sigma(x)$
in which it resides. A directed arc connects a gene $x$ with its closest
relatives in each of the other species \cite{Geiss:2019a}. The
underlying graph of a BMG that contains only bi-directional arcs, that is,
those arcs $(x,y)$ for which there is also an arc $(y,x)$, is known as
reciprocal best match graph (RBMG). The precise definition of BMGs will be
given in Section \ref{sect:prelim}. Empirically, best matches are
routinely estimated by measuring and comparing the similarity of gene
sequences. Measurement errors and systematic biases, however, introduce
discrepancies between ``most similar genes'' extracted from data and the
notion of best matches in the sense of closest evolutionary relatedness
\cite{Geiss:2019a,Geiss:2020c}. While some systematic effects can be
corrected directly \cite{Stadler:2020}, a residual level of error is
unavoidable. It is therefore a question of considerable practical interest
in computational biology whether the mathematical properties characterizing
BMGs can be used to correct empirical estimates. Formally, this question
amounts to a graph editing problem: Given a vertex-labeled directed graph
$(\G,\sigma)$, what is the minimal number of arcs that need to be inserted
or deleted to convert $(\G,\sigma)$ into a BMG $(\G^*,\sigma)$?

\begin{figure}[t]
  \begin{center}
    \includegraphics[width=0.85\linewidth]{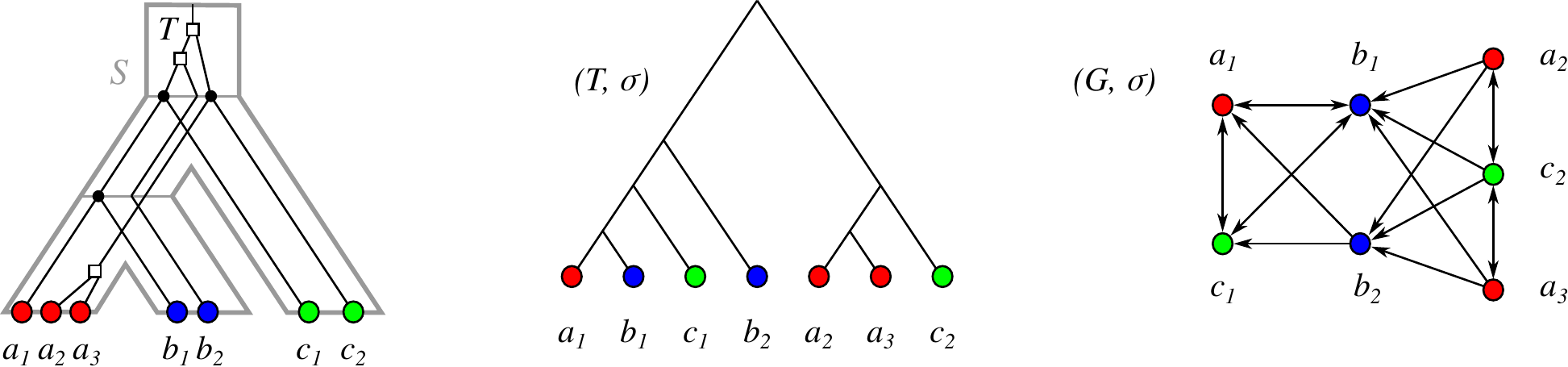}
  \end{center}
  \caption[]{An evolutionary scenario (left) consists of a gene tree
    $(T,\sigma)$ embedded into a species phylogeny $S$. The coloring
    $\sigma$ represents the species, i.e., genomes (leaves of the species
    tree) in which the genes (leaves of the gene tree) reside.  Inner nodes of
    the gene tree correspond to gene duplications ($\DUPL$) or speciation
    events ($\SPEC$), the latter coinciding with the inner nodes of the
    species tree $S$.  The BMG $(\G,\sigma)$ (right) depends on the
    topology of the gene tree (middle). A gene $y$ is a best match of $x$
    ($x\rightarrow y$ in the BMG) if there is no gene $y'$ of the same
    color that has a lower last common ancestor with $x$ than $y$. For
    instance, $b_1$ but not $b_2$ is a best match for $a_1$ in the blue
    species.}
  \label{fig:bmg_example}
\end{figure}

Best matches are, in particular, closely linked to the identification of
orthologous genes, i.e., pairs of genes whose last common ancestor
coincides with the divergence of two species \cite{Fitch:70}. Orthologous
genes from different species are expected to have essentially the same
biological functions. Thus, considerable efforts have been expended to
devise methods for orthology assessment, see e.g.\
\cite{Sonnhammer:14,Altenhoff:16,Setubal:18a} for reviews and
applications. The orthology graph of a gene family (with the genes as
vertices and undirected edges between orthologous genes) can be shown to be
a subgraph of the reciprocal best match graph (RBMG), i.e., the symmetric
part of the BMG \cite{Geiss:2020c}. This has sparked interest in a
characterization of RBMGs \cite{Geiss:2020b} and the corresponding graph
editing problems \cite{Hellmuth:2020a}. The deletion and the editing
problems of 2-colored RBMGs are equivalent to \PROBLEM{Bicluster Deletion}
and \PROBLEM{Bicluster Editing}, respectively, a fact that was used to
demonstrate NP-hardness for the general, $\ell$-colored case. On the other
hand, orthology graphs are cographs \cite{Hellmuth:13a}. \PROBLEM{Cograph
  Editing} or \PROBLEM{Cograph Deletion} thus have been used to correct
empirical approximations of RBMGs to orthology graphs in
\cite{Hellmuth:15}.  Several related problems have been discussed in
  the literature, often aiming at using additional biological information
  as part of the graph modification process, cf.\ eg.\
  \cite{Natanzon:2001,lafond2013gene,lafond2014orthology,
    lafond2016link,dondi2017approximating}.  Both \PROBLEM{Cograph
  Editing} and \PROBLEM{Cograph Deletion} are NP-complete
\cite{Liu:2012}. In \cite{Schaller:2020}, we showed that knowledge of BMG
makes it possible to identify the edges of the RBMG that cannot be part of
the orthology graph and found that these edges in general do not form an
optimal solution of either \PROBLEM{Cograph Editing} or \PROBLEM{Cograph
  Deletion}. This observation suggests to correct the empirical similarity
data at the outset by editing them to the nearest BMGs instead of operating
on an empirical approximation of the RBMG. Given a BMG, the orthology graph
can then be computed in polynomial time \cite{Schaller:2020}.

We therefore analyze the arc modification problems for $\ell$-BMGs, that
is, BMGs on $\ell$ colors.  This contribution is organized as follows:
After introducing some notation and reviewing some important properties of
BMGs, Sec.~\ref{sec:triples} provides a characterization of BMGs in
  terms of triples (binary trees on three leaves) that can be derived from
  vertex colored graphs.  Moreover, we provide in
  Sec.~\ref{sect:forbidden} a characterization of 2-BMGs in terms of
forbidden subgraphs. We then prove in Sec.~\ref{sect:NPcomplete2} that
\PROBLEM{$2$-BMG Deletion} and \PROBLEM{$2$-BMG Editing} are NP-complete by
reduction from \PROBLEM{Exact 3-Cover}, and that \PROBLEM{$2$-BMG
Completion} is NP-complete by reduction from \PROBLEM{Chain Graph
Completion}. These results are used in Sec.~\ref{sect:NPcomplete-all} to
establish NP-completeness for any fixed number $\ell\ge2$ of colors.
Finally, we provide ILP solutions for the respective $\ell$-BMG
modification problems in Sec.~\ref{sect:ILP}.

\section{Preliminaries}
\label{sect:prelim}

\subsection{Notation}

In this contribution, we consider simple directed graphs (digraphs)
$\G=(V,E)$ with vertex set $V$ and arc set
$E\subseteq V\times V \setminus \{(v,v)\mid v\in V\}$.
We will also write $V(\G)$ and $E(\G)$ when referring to the vertex and 
arc set, respectively, of a specific graph $\G$.

For a vertex $x\in V$, we say that $(y,x)$ is an \emph{in-arc} and $(x,z)$ is
an \emph{out-arc}. The (weakly) connected components of $\G$ are the maximal
connected subgraphs of the undirected graph underlying $\G$. We call $x$ a
\emph{hub-vertex} of a graph $\G=(V,E)$ if $(x,v)\in E$ and $(v,x)\in E$
holds for all vertices $v\in V\setminus \{x\}$. The subgraph induced by a
subset $W\subseteq V$ is denoted by $\G[W]$. We write
$N(x):=\{z\in V\mid (x,z)\in E\}$ for the out-neighborhood and
$N^-(x):=\{z\in V\mid (z,x)\in E\}$ for the in-neighborhood of $x\in V$.  A
graph is \emph{sink-free} if it has no vertex with out-degree zero, i.e.,
if $N(x)\ne\emptyset$ for all $x\in V$.

We write $E\symdiff F \coloneqq (E\setminus F) \cup (F\setminus E)$ for the
symmetric difference of the sets $E$ and $F$. Moreover, for a graph
$\G=(V,E)$ and an arc set $F$, we define the graphs
$\G+F\coloneqq (V, E\cup F)$, $\G-F\coloneqq (V, E\setminus F)$ and
$\G\symdiff F\coloneqq (V, E\symdiff F)$.  A \emph{vertex coloring} of $\G$ 
is a
map $\sigma: V\to M$, where $M$ is the set of 
\emph{colors}.  A graph $\G$ endowed with a vertex coloring $\sigma$ will 
be denoted by $(\G,\sigma)$.  A vertex coloring is \emph{proper} if
$\sigma(x)\ne\sigma(y)$ for all $(x,y)\in E$.  To avoid trivial cases,
  we will always assume a coloring $\sigma$ to be surjective. To this end,
  we define the restriction of $\sigma$ to a subset $W\subseteq V$ of
vertices as $\sigma_{|W}\colon W\to \sigma(W)$ with
  $\sigma_{|W}(v)=\sigma(v)$ for all $v\in W$. The colored subgraph of
$\G$ induced by $W$ is therefore $(\G[W], \sigma_{|W})$.  We often write
$|M|$-coloring to emphasize the number of colors in $\G$. Moreover,
we define $\sigma(W):=\{\sigma(v)\mid v\in W\}$.
\begin{fact}
  Let $x$ be a hub-vertex in a properly colored graph $(\G,\sigma)$.  Then
  $x$ is the only vertex of color $\sigma(x)$ in $(\G,\sigma)$.
\end{fact}

A phylogenetic tree $T$ (on $L$) is an (undirected) rooted tree with root
$\rho_T$, leaf set $L=L(T)\subseteq V(T)$ and inner vertices
$V^0(T) = V(T) \setminus L$ such that each inner vertex of T (except
possibly the root) is of degree at least three.  \emph{Throughout this
  contribution, we assume that every tree is phylogenetic.}

The \emph{ancestor order} on $V(T)$ is defined such that $u\preceq_T v$ if
$v$ lies on the unique path from $u$ to the root $\rho_T$, i.e., if $v$ is
an ancestor of $v$. We write $u \prec_T v$ if $u \preceq_{T} v$ and
$u \neq v$. If $xy$ is an edge in $T$, such that $y \prec_{T} x$, then $x$
is the \emph{parent} of $y$ and $y$ the \emph{child} of $x$. We denote by
$\child_T(x)$ the set of all children of $x$. The set $L(T(v))$
  contains of all leaves $x\preceq_T v$. For a non-empty subset
$A\subseteq V\cup E$ we define $\lca_T(A)$, the \emph{last common ancestor
  of $A$}, to be the unique $\preceq_T$-minimal vertex of $T$ that is an
ancestor of every $u\in A$.  For simplicity we write
$\lca_T(A)=\lca_T(x_1,\dots,x_k)$ whenever we specify a vertex set
$A=\{x_1,\dots,x_k\}$ explicitly.  Note that $\lca_T(x,y)$ and
$\lca_T(x,z)$ are comparable for all $x,y,z\in L$ w.r.t.\ $\preceq_T$.

A \emph{(rooted) triple} is a tree on three leaves and with two inner
vertices.  We write $xy|z$ for the triple on the leaves $x,y$ and $z$ if
the path from $x$ to $y$ does not intersect the path from $z$ to the root,
i.e., if $\lca_T(x,y)\prec_T \lca_T(x,z)=\lca_T(y,z)$. In this case we say
that $T$ displays $xy|z$. A set $\mathscr{R}$ of triples on $L$, i.e., a
set of triples $\mathscr{R}$ such that $\bigcup_{T\in\mathscr{R}} L(T)=L$,
is \emph{compatible} if there is a tree with leaf set $L$ that displays
every triple in $L$. If $\mathscr{R}$ is compatible, then such a tree, the
\emph{Aho tree} $\Aho(\mathscr{R})$ can be constructed in polynomial time
\cite{Aho:81}.  For a set $L$, a set of triples $\mathscr{R}$ is
\emph{strictly dense} if for all three distinct $x,y,z$ exactly one of the
triples $xy|z, xz|y$ and $yz|x$ is contained in $\mathscr{R}$.

In this contribution, we also consider trees that explicitly do not
display certain triples. More precisely, we will need
\begin{definition}
  Let $\mathscr{R}$ and $\mathscr{F}$ be sets of triples. The pair
  $(\mathscr{R},\mathscr{F})$ is called \emph{compatible} if there is a
  tree $T$ that displays all triples in $\mathscr{R}$ but none of the
  triples in $\mathscr{F}$. In the latter case, we also say that $T$
  \emph{agrees with} $(\mathscr{R},\mathscr{F})$.
\end{definition}
\noindent
The problem of deciding whether such a pair $(\mathscr{R},\mathscr{F})$ is
compatible and, in particular, of finding a corresponding phylogentic tree
was termed \emph{mixed triplets problem restricted to trees (MTT)} in
\cite{He:06}.  This work also reports a polynomial-time algorithm (also
called \texttt{MTT}), which is similar to the well-known \texttt{BUILD}
algorithm \cite{Aho:81}.
\begin{theorem}{\cite[Thm.~1]{He:06}}
  Algorithm \texttt{MTT} outputs a phylogenetic tree $T$ that agrees with the
  pair $(\mathscr{R},\mathscr{F})$ (defined on $n$ distinct leaves) if 
  and only if $(\mathscr{R},\mathscr{F})$
  is compatible in $O(|\mathscr{R}|\cdot n+|\mathscr{F}|\cdot n \log
  n + n^2 \log n)$ time.
\end{theorem}
A tree $T$ with leaf set $L$ together with function
$\sigma:L\to M$ is a \emph{leaf-colored tree}, denoted by
$(T,\sigma)$.

\subsection{Best match graphs}

\begin{definition}
  Let $(T,\sigma)$ be a leaf-colored tree. A leaf $y\in L(T)$ is a
  \emph{best match} of the leaf $x\in L(T)$ if $\sigma(x)\neq\sigma(y)$ and
  $\lca(x,y)\preceq_T \lca(x,y')$ holds for all leaves $y'$ of color
  $\sigma(y')=\sigma(y)$.
\end{definition}
The graph $\G(T,\sigma) = (V,E)$ with vertex set $V=L(T)$, vertex coloring
$\sigma$, and with arcs $(x,y)\in E$ if and only if $y$ is a best match of
$x$ w.r.t.\ $(T,\sigma)$ is known as the (colored) \emph{best match graph}
(BMG) of $(T,\sigma)$ \cite{Geiss:2019a}. We call an $\ell$-colored BMG
simply $\ell$-BMG.  Since the last common ancestors of any two vertices of
$T$ always exists, and $\lca_T(x,y)$ and $\lca_T(x,z)$ are comparable,
there is by definition at least one best match of $x$ for every color
$s\in\sigma(V)\setminus\{\sigma(x)\}$:
\begin{fact}
  For every vertex $x$ and every color
  $s\in\sigma(V(G))\setminus\{\sigma(x)\}$ in a BMG $(\G,\sigma)$,
  there is some vertex $y\in N(x)$ with $\sigma(y)=s$.
  \label{fact:allcolors-out}
\end{fact}
Equivalently, the subgraph induced by every pair of colors is
sink-free. In particular, therefore, BMGs are sink-free whenever they
contain at least two colors. We note in passing that sink-free graphs also
appear naturally e.g.\ in the context of graph semigroups \cite{Abrams:10}
and graph orientation problems \cite{Cohn:02}.

\begin{definition}
  Let $(G,\sigma)$ be a colored graph. The coloring $\sigma$ is
  \emph{sink-free} if it is proper and, for every vertex $x$ and every
  color $s\in\sigma(V(G))\setminus\{\sigma(x)\}$, 
  there is a vertex $y\in N(x)$ with $\sigma(y)=s$.
  A graph with a sink-free coloring is \emph{sf-colored}.
\end{definition}

The definition of BMGs together with Obs.~\ref{fact:allcolors-out} implies
that BMGs are always sf-colored.

\begin{definition}\label{def:BestMatchGraph}
  An arbitrary vertex-colored graph $(\G,\sigma)$ is a \emph{best match
    graph (BMG)} if there exists a leaf-colored tree $(T,\sigma)$ such that
  $(\G,\sigma) = \G(T,\sigma)$. In this case, we say that $(T,\sigma)$
  \emph{explains} $(\G,\sigma)$.
\end{definition}
Whether two vertices $x$ and $y$ are best matches or not does not depend on
the presence or absence of vertices $z$ with
$\sigma(z)\notin\{\sigma(x),\sigma(y)\}$. More precisely, we have
\begin{fact}{\cite[Obs.~1]{Geiss:2019a}}
  \label{obs:color-restriction}
  Let $(\G,\sigma)$ be a BMG explained by $(T,\sigma)$ with leaf set $L$
  and let $L'\coloneqq \bigcup_{s\in M'} L[s]$ be a subset
  of vertices with a restricted color set
  $M'\subseteq \sigma(L)$.  Then the induced subgraph
  $(\G[L'],\sigma_{|L'})$ is explained by the restriction $T_{|L'}$ of $T$
  to the leaf set $L'$,
  i.e. $(\G[L'],\sigma_{|L'})=\G(T_{|L'},\sigma_{|L'})$.
\end{fact}

It was shown in \cite{Geiss:2019a} that BMGs can be characterized in terms
of certain induced subgraphs on three vertices. These can be specificed
as follows \cite{Schaller:2020}:
\begin{definition}\label{def:informative_triples}
  Let $(\G,\sigma)$ be a  vertex colored graph. We say that a
  triple $xy|y'$ is \emph{informative} for $(\G,\sigma)$ if $x$, $y$ and
  $y'$ are pairwise distinct vertices in $\G$ such that (i)
  $\sigma(x)\neq\sigma(y)=\sigma(y')$ and (ii) $(x,y)\in E(\G)$ and
  $(x,y')\notin E(\G)$. The set of informative triples is denoted by
  $\mathscr{R}(\G,\sigma)$.
\end{definition}

As shown in \cite{BMG-corrigendum}, BMGs can be characterized in
terms of informative triples. 
\begin{theorem}{\cite[Thm.~15]{BMG-corrigendum}}
  A colored digraph $(\G,\sigma)$ is an $n$-cBMG if and only if
  $\G(\Aho(\mathscr{R}(\G,\sigma)),\sigma) = (\G,\sigma)$. 
  \label{thm9_new}
\end{theorem}
However, for our purposes, it will be convenient to avoid the
construction of the Aho tree. In Sec.~\ref{sec:triples}, we
will establish an alternative characterization, which will depend on
the following two
technical results:
\begin{lemma}{\cite[Lemma~5 and~6]{Schaller:2020}}
  \label{lem:informative_triples}
  Let $(\G,\sigma)$ be a BMG and $xy|y'$ an informative triple for $\G$.
  Then, every tree $(T,\sigma)$ that explains $(\G,\sigma)$ displays the
  triple $xy|y'$, i.e. $\lca_T(x,y)\prec_T\lca_T(x,y')=\lca_T(y,y')$.

  Moreover, if the triples $ab|b'$ and $cb'|b$ are informative for
  $(\G,\sigma)$, then every tree $(T,\sigma)$ that explains $(\G,\sigma)$
  contains two distinct children $v_1,v_2\in \child_T(\lca_T(a,c))$ such
  that $a,b\prec_T v_1$ and $b',c\prec_T v_2$.
\end{lemma}

\begin{lemma}{\cite[Prop.\ 1]{Geiss:2019a}}
  The disjoint union of vertex disjoint BMGs $(\G_i, \sigma_i)$,
  $1\leq i\leq k$, is a BMG if and only if all color sets are the same,
  i.e., $\sigma_i(V(\G_i)) = \sigma_j(V(\G_j))$ for $1\leq i< j\leq k$.
\label{lem:BMGunion}
\end{lemma}

\subsection{Neighborhoods in BMGs}

A graph is \emph{thin} if no two vertices have the same neighborhood.
\begin{definition}\label{def:rthin}
  Two vertices $x,y\in L$ are in relation $\rthin$ if $N(x)=N(y)$ and
  $N^-(x)=N^-(y)$.
\end{definition}
Clearly the \emph{thinness relation} $\rthin$ is an equivalence relation on
$V$.  For each $\rthin$ class $\alpha$ we have $N(\alpha)=N(x)$ and
$N^-(\alpha) = N^-(x)$ for all $x\in\alpha$.  The following
  characterization of 2-BMGs makes use of the structure of the trees by
  which they are explained. These properties can be expressed in terms of
  properties of the vertex neighborhoods in the 2-BMGs. While they can be
  tested in polynomial time, they are far from being intuitive. We refer to
  \cite{Geiss:2019a} for a detailed discussion.
\begin{theorem}{\cite[Thm.~3 and 4]{Geiss:2019a}}
  \label{thm:charact-2BMG} Let $(\G,\sigma)$ be a connected properly
  2-colored digraph.  Then, $(\G,\sigma)$ is a BMG if and only if for any
  two $\rthin$ classes $\alpha$ and $\beta$ of $\G$ holds
  \begin{itemize}
    \setlength{\itemindent}{0.15in}
  \item[\AX{(N0)}] $N(\alpha)\ne\emptyset$
  \item[\AX{(N1)}] $\alpha\cap N(\beta)=\beta\cap N(\alpha)=\emptyset$
    implies \\
    $N(\alpha) \cap N(N(\beta))=N(\beta)  \cap N(N(\alpha))=\emptyset$.
  \item[\AX{(N2)}] $N(N(N(\alpha))) \subseteq N(\alpha)$
  \item[\AX{(N3)}]
    $\alpha\cap N(N(\beta))=\beta\cap N(N(\alpha))=\emptyset$ and
    $N(\alpha)\cap N(\beta)\neq \emptyset$ implies
    $N^-(\alpha)=N^-(\beta)$ and
    $N(\alpha)\subseteq N(\beta)$ or
    $N(\beta)\subseteq N(\alpha)$.
  \end{itemize}
\end{theorem}
We note that \cite{Geiss:2019a} tacitly assumed \AX{(N0)}, i.e., that
$(\G,\sigma)$ is sink-free.

\section{Characterization of BMGs in terms of triples}
\label{sec:triples}

In this section, we provide a novel characterization of BMGs utilizing
allowed and forbidden triples.  To this end, we need
\begin{definition}\label{def:forbidden_triples}
  Let $(\G,\sigma)$ be a vertex colored graph. We say that a triple $xy|y'$
  is \emph{forbidden} for a graph $(\G,\sigma)$ if $x$, $y$ and $y'$ are
  pairwise distinct vertices in $\G$ such that (i)
  $\sigma(x)\neq\sigma(y)=\sigma(y')$ and (ii) $(x,y)\in E(\G)$ and
  $(x,y')\in E(\G)$. The set of forbidden triples of $(\G,\sigma)$ is
  denoted by $\mathscr{F}(\G,\sigma)$.
\end{definition}
The characterization of BMGs will make use of the following two
technical results:
\begin{lemma}\label{lem:forb_triples}
  Let $(\G,\sigma)$ be a BMG explained by $(T,\sigma)$.  Then, none of the
  triples in $\mathscr{F}(\G,\sigma)$ is displayed by $(T,\sigma)$.
\end{lemma}
\begin{proof}
  Assume, for contradiction, that $(T,\sigma)$ displays
  $xy|y'\in \mathscr{F}(\G,\sigma)$.  Hence, $\lca_T(x,y)\prec_T \lca(x,y')$
  and thus, $(x,y)\in E(\G)$ and $(x,y')\notin E(\G)$ contradicting
  the definition of $xy|y'$ as a forbidden triple of $(G,\sigma)$.
\end{proof}

\begin{lemma}\label{lem:subgraph-Aho}
  Let $(\G,\sigma)$ be an sf-colored graph
  with vertex set $L$.
  Then for every tree $(T,\sigma)$ on $L$ displaying all triples in
  $\mathscr{R}(\G,\sigma)$, it holds $E(\G(T,\sigma))\subseteq E(\G)$.
\end{lemma}
\begin{proof}
  Let $(T,\sigma)$ be a tree displaying all triples in
  $\mathscr{R}(\G,\sigma)$ and set $E'\coloneqq E(\G(T,\sigma))$ and
  $E\coloneqq E(\G)$.  First note that $(\G,\sigma)$ and $\G(T,\sigma)$
  have the same vertex set $L$.  Suppose that $(a, b)\in E'$ but
  $(a, b)\notin E$.  Since $(G,\sigma)$ is sf-colored, vertex $a$
  must have at least one out-neighbor $b'$ (distinct from $b$) of color
  $\sigma(b)$ in $(\G,\sigma)$, i.e. $(a,b')\in E$.  Hence, $(a,b')\in E$
  and $(a, b)\notin E$ imply that $ab'|b$ is an informative triple of
  $(\G,\sigma)$ and thus displayed by $T$.  Therefore
  $\lca_T(a,b')\prec_T\lca_T(a,b)$ which, together with
  $\sigma(b)=\sigma(b')$, contradicts $(a,b)\in E'$.  Therefore,
  $E'\subseteq E$.
\end{proof}

\begin{proposition}\label{prop:explaining-tree}
  Let $(\G,\sigma)$ be a sf-colored graph with vertex set $L$. A
  leaf-colored tree $(T,\sigma)$ on $L$ explains $(\G,\sigma)$ if and only
  if $(T,\sigma)$ agrees with
  $(\mathscr{R}(\G,\sigma), \mathscr{F}(\G,\sigma) )$.  In this case,
  $(\G,\sigma)$ is a BMG.
\end{proposition}
\begin{proof}
  First suppose that $(T,\sigma)$ explains $(\G,\sigma)$, in which case
  $(\G,\sigma)$ is a BMG.  The only-if-direction now immediately follows
  from Lemmas~\ref{lem:informative_triples} and~\ref{lem:forb_triples}.

  Now suppose that there is a tree $(T,\sigma)$ on $L$ that displays all
  triples in $\mathscr{R}(\G,\sigma)$ and none of the triples in
  $\mathscr{F}(\G,\sigma)$.  Hence, we can apply
  Lemma~\ref{lem:subgraph-Aho} to conclude that
  $E'\coloneqq E(\G(T,\sigma)) \subseteq E(\G)\eqqcolon E$.  Note that
  $(\G,\sigma)$ and $\G(T,\sigma)$ have the same vertex set $L$.  We show
  that $E'=E$. Assume, for contradiction, that $E'\subset E$, and thus,
  that there is an $(a,b)\in E\setminus E'$.  By
  Obs.~\ref{fact:allcolors-out} and since $\G(T,\sigma)$ is a BMG, vertex
  $a$ must have at least one out-neighbor $b'$ of color $\sigma(b)$.
  Hence, there is an arc $(a,b')\in E'$. Thus, $ab'|b$ is an informative
  triple of $\G(T,\sigma)$ and must therefore be displayed by $T$.
  Moreover, $(a,b')\in E'$ and $E'\subset E$ imply $(a,b')\in E$.  Hence,
  $(a,b), (a,b')\in E$ implies that $ab'|b$ is a forbidden triple of
  $(\G,\sigma)$ and thus, not displayed by $T$ by assumption; a
  contradiction.  Therefore, $E=E'$ must hold and thus,
  $\G(T,\sigma) = (\G,\sigma)$ which, in particular, implies that
  $(\G,\sigma)$ is a BMG.
\end{proof}

\begin{theorem}\label{thm:BMG-charac-via-R-F}
  A vertex colored graph $(\G,\sigma)$ is a BMG if and only if (i)
  $(\G,\sigma)$ is sf-colored and (ii)
  $(\mathscr{R}(\G,\sigma),\mathscr{F}(\G,\sigma))$ is compatible.
\end{theorem}
\begin{proof}
  First suppose that $(\G,\sigma)$ is a BMG.  The definition of BMGs
  together with Obs.\ \ref{fact:allcolors-out} implies that $(\G,\sigma)$
  is sf-colored.  Thus, there is a tree $(T,\sigma)$ that explains
  $(\G,\sigma)$.  By Lemmas~\ref{lem:informative_triples}
  and~\ref{lem:forb_triples}, $(T,\sigma)$ displays all triples in
  $\mathscr{R}(\G,\sigma)$ and none of the triples in
  $\mathscr{F}(\G,\sigma)$.  Hence, the pair
  $(\mathscr{R}(\G,\sigma),\mathscr{F}(\G,\sigma))$ is compatible.
  
  For the converse, suppose that $(\G,\sigma)$ is a sf-colored
  graph and that $(\mathscr{R}(\G,\sigma),\mathscr{F}(\G,\sigma))$ is
  compatible.  The latter implies that there is a tree $(T,\sigma)$ on $L$
  displaying all triples in $\mathscr{R}(\G,\sigma)$ and none of the
  triples in $\mathscr{F}(\G,\sigma)$.  Now, we can apply
  Prop.~\ref{prop:explaining-tree} to conclude that $(\G,\sigma)$ is a BMG.
\end{proof}

In order to use the \texttt{MTT} algorithm \cite{He:06} to recognize 
BMGs $(\G,\sigma)$, we show for completeness that the set of allowed and 
forbidden triples already determines $V(\G)$ except for trivial cases.

\begin{lemma}\label{lem:no-vertex-lost}
  Let $(\G,\sigma)$ be a sf-colored graph, $V(\G)\neq \emptyset$ and 
  $L'\coloneqq\bigcup_{t\,\in\, 
  \mathscr{R}(\G,\sigma)\,\cup\,\mathscr{F}(\G,\sigma)} L(t)$.
  Then the following statements are equivalent:
  \begin{enumerate}[noitemsep]
  \item $L'=V(\G)$
  \item $\mathscr{R}(\G,\sigma)\cup\mathscr{F}(\G,\sigma)\ne\emptyset$
  \item $(\G,\sigma)$ is $\ell$-colored with $\ell\geq 2$ and contains
    two vertices of the same color.
  \end{enumerate}
  \noindent
  Otherwise, $(\G,\sigma)$ is a BMG that is explained by any tree
  $(T,\sigma)$ on $V(G)$.
\end{lemma}
\begin{proof}
  The fact that $L'=V(\G)\neq \emptyset$ immediately implies that
  $\mathscr{R}(\G,\sigma)\cup\mathscr{F}(\G,\sigma)$ must not be empty.
  Hence, (1) implies~(2).

  Suppose Condition~(2) is satisfied.  Since all triples
  $xy|y'\in \mathscr{R}(\G,\sigma)\cup\mathscr{F}(\G,\sigma)$ satisfy
  $\sigma(x)\ne\sigma(y)=\sigma(y')$, Condition (3) must be satisfied.
  Hence, (2) implies~(3).

  Suppose Condition~(3) is satisfied. Hence, there are two vertices of the
  same color $r$ and there must be a vertex $x\in V(\G)$ with
  $\sigma(x)\ne r$.  Since $(\G,\sigma)$ is sf-colored, there is a
  vertex $y\in V(\G)$ of color $r$ such that $(x,y)\in E(\G)$. Now let
  $y'\in V(\G)$, $y'\neq y$ be one of the additional vertices of color
  $\sigma(y')=r$. If $(x,y')\notin E(\G)$ then
  $xy|y'\in\mathscr{R}(\G,\sigma)$ and, otherwise, if $(x,y')\in E(\G)$
  then $xy|y'\in\mathscr{F}(\G,\sigma)$.  In summary, every vertex $x$ of
  $(\G,\sigma)$ is part of some informative or forbidden triple and thus,
  $L'=V(\G)$.  Hence, (3) implies~(1).

  Finally, suppose that none of the equivalent statements~(1), (2),
  and~(3) holds. Then $(\G,\sigma)$ is either
  1-colored and thus, does not contain any arc, or $|V(\G)|$-colored in
  which case $(\G,\sigma)$ is a complete graph. In both cases, the tree
  topology of $(T,\sigma)$ does not matter.
\end{proof}

It is straightforward to test whether a vertex colored graph
$(G,\sigma)$ is sf-colored in $O(|E(G)|)$ time.  Moreover,
\texttt{MTT} \cite{He:06} accomplishes the compatibility check of
$(\mathscr{R}(\G,\sigma),\mathscr{F}(\G,\sigma))$ and the construction of
a corresponding tree in polynomial time. If
$\mathscr{R}(\G,\sigma)\cup\mathscr{F}(\G,\sigma)=\emptyset$,
$(\G,\sigma)$ is a BMG. Otherwise, Lemma~\ref{lem:no-vertex-lost} implies
that every vertex in the sf-colored graph $(G,\sigma)$ appears
in an informative and/or a forbidden triple. Together with
Prop.~\ref{prop:explaining-tree} and Thm.~\ref{thm:BMG-charac-via-R-F}, this 
implies
\begin{corollary}\label{cor:bmg-rec-polytime}
  It can be determined in polynomial time whether a vertex colored
  graph $(G,\sigma)$ is a BMG. In the affirmative case, a tree that
  explains $(G,\sigma)$ can constructed in polynomial time.
\end{corollary}

\section{Forbidden induced subgraphs of 2-BMGs}
\label{sect:forbidden}

In this section, we derive a new characterization of 2-colored BMGs in terms
of forbidden induced subgraphs. Our starting point is the observation that
certain constellations of arcs on four or five vertices cannot occur.

\begin{definition}[F1-, F2-, and F3-graphs]\par\noindent
\begin{itemize}
\item[\AX{(F1)}] A properly 2-colored graph on four distinct vertices
  $V=\{x_1,x_2,y_1,y_2\}$ with coloring
  $\sigma(x_1)=\sigma(x_2)\ne\sigma(y_1)=\sigma(y_2)$ is an
  \emph{F1-graph} if $(x_1,y_1),(y_2,x_2),(y_1,x_2)\in
  E$ and $(x_1,y_2),(y_2,x_1)\notin E$.
\item[\AX{(F2)}] A properly 2-colored graph on four distinct vertices
  $V=\{x_1,x_2,y_1,y_2\}$ with coloring
  $\sigma(x_1)=\sigma(x_2)\ne\sigma(y_1)=\sigma(y_2)$ is an
  \emph{F2-graph} if $(x_1,y_1),(y_1,x_2),(x_2,y_2)\in E$ and
  $(x_1,y_2)\notin E$.
\item[\AX{(F3)}] A properly 2-colored graph on five distinct vertices
  $V=\{x_1,x_2,y_1,y_2,y_3\}$ with coloring
  $\sigma(x_1)=\sigma(x_2)\ne\sigma(y_1)=\sigma(y_2)=\sigma(y_3)$ is an
  \emph{F3-graph} if $(x_1,y_1),(x_2,y_2),(x_1,y_3),(x_2,y_3)\in E$ and
  $(x_1,y_2),(x_2,y_1)\notin E$.
\end{itemize}
\label{def:forbidden-subgraphs}
\end{definition}

\begin{figure}[t]
  \begin{center}
    \includegraphics[width=0.85\linewidth]{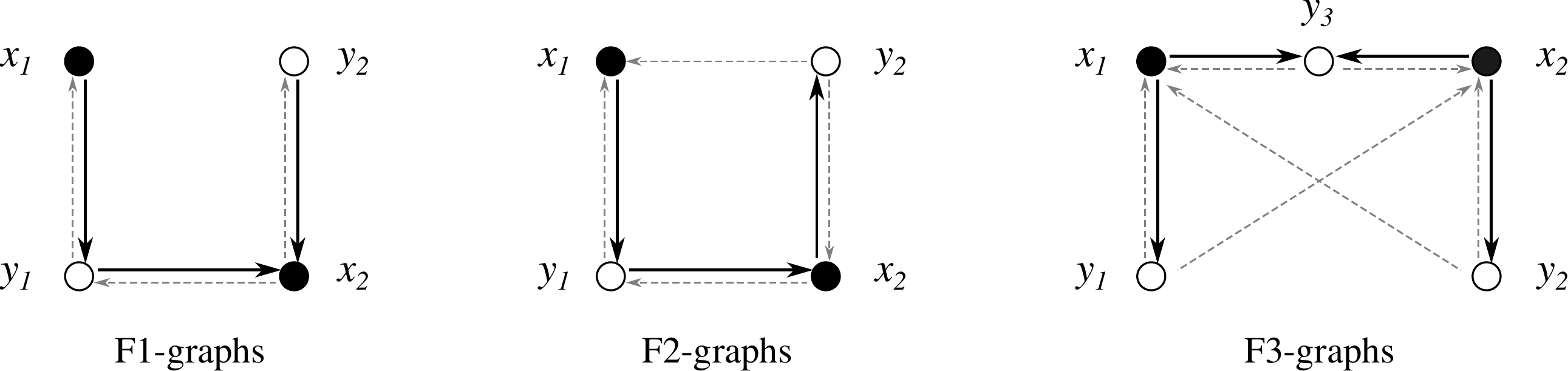}
  \end{center}
  \caption[]{Templates of the three families of forbidden induced subgraphs
    in BMGs. Black arcs must exist, non-arcs must not exist and dashed gray arcs may or may not be
    present.}
  \label{fig:forbidden-subgraphs}
\end{figure}

The ``templates'' for F1-, F2-, and F3-graphs are shown in
Fig.~\ref{fig:forbidden-subgraphs}. They define 8, 16, and 64 graphs by
specifying the presence or absence of the 3, 4, and 6 optional (dashed)
arcs, respectively, see Figs.~\ref{fig:all-F1-F2} and
\ref{fig:all-F3} in the Appendix. The F1- and
F2-graphs fall into a total of 16 isomorphism classes, four of which are
both F1- and F2-graphs. All but one of the F3-graphs contain an F1- or
an F2-graph as induced subgraph. The exception is the ``template'' of
the F3-graphs without optional arcs. The 17 non-redundant forbidden
subgraphs are collected in
Fig.~\ref{fig:all-17}. We shall see below that
they are sufficient to characterize 2-BMGs among the sink-free graphs.

\begin{figure}[t]
  \begin{center}
\includegraphics[width=0.85\linewidth]{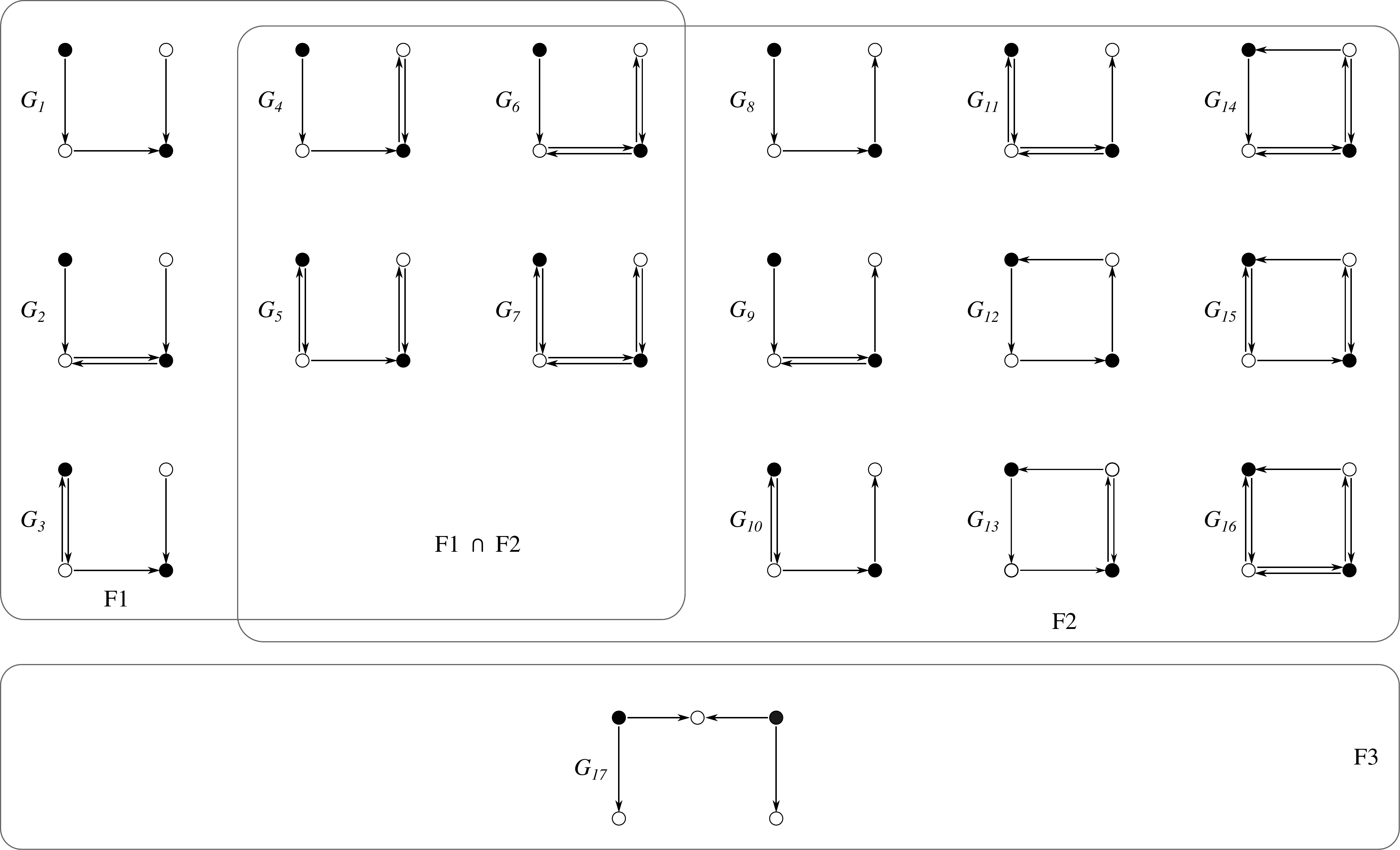}
  \end{center}
  \caption[]{Forbidden induced subgraphs in BMGs. All F3-graphs
      with at least one optional arc have an induced F1- or
      F2-graph and thus are redundant.}
  \label{fig:all-17}
\end{figure}

\begin{lemma}
  \label{lem:forbidden-subgraphs}
  If $(\G,\sigma)$ is a BMG, then it contains no induced  F1-, F2-, or
  F3-graph.
\end{lemma}
\begin{proof}
  Let $(T,\sigma)$ be a tree that explains $(\G,\sigma)$.

  First, assume that $(\G,\sigma)$ contains an induced F1-graph, i.e.,
  there are four vertices $x_1,x_2,y_1,y_2$ satisfying \AX{(F1)}, and let
  $u\coloneqq \lca_T(x_1,y_2)$.  Then, $(x_1,y_1),(y_2,x_2)\in E$,
  $(x_1,y_2),(y_2,x_1)\notin E$ and Lemma \ref{lem:informative_triples}
  imply that $T$ must display the informative triples $x_1y_1|y_2$ and
  $y_2x_2|x_1$.  Hence, $u$ must have two distinct children $v_1$ and $v_2$
  such that $x_1,y_1\prec_T v_1$ and $x_2,y_2\prec_T v_2$.  Therefore,
  $\lca_T(x_1,y_1)\preceq_T v_1 \prec_T u =\lca_T(x_2,y_1)$ and
  $\sigma(x_1)=\sigma(x_2)$ imply that $(y_1,x_2)\notin E(\G)$; a
  contradiction.

  Next, assume that $(\G,\sigma)$ contains an induced F2-graph, i.e., there
  are four vertices $x_1,x_2,y_1,y_2$ satisfying \AX{(F2)}. Then
  $(x_1,y_1)\in E$, $(x_1,y_2)\notin E$ and
  Lemma~\ref{lem:informative_triples} imply that $T$ displays the
  informative triple $x_1y_1|y_2$ and thus
  $\lca_T(x_1,y_1)\prec_{T}\lca_T(x_1,y_2)$.  Since $(y_1,x_2)\in E$ and
  $\sigma(x_1)=\sigma(x_2)$, we conclude that
  $\lca_T(x_2,y_1)\preceq_T\lca_T(x_1,y_1)\prec_{T}\lca_T(x_1,y_2)$ and
  therefore also $\lca_T(x_2,y_1)\prec_T\lca_T(x_2,y_2)=\lca_T(x_1,y_2)$.
  Together with $\sigma(y_1)=\sigma(y_2)$, the latter contradicts
  $(x_2,y_2)\in E$.

  Finally, assume that $(\G,\sigma)$ contains an induced F3-graph, i.e.,
  there are five vertices $x_1,x_2,y_1,y_2,y_3$ satisfying \AX{(F3)}. By
  Lemma~\ref{lem:informative_triples}, $(x_1,y_1)\in E$ and
  $(x_1,y_2)\notin E$ implies that $T$ displays the triple $x_1y_1|y_2$,
  and $(x_2,y_2)\in E$ together with $(x_2,y_1)\notin E$ implies that $T$
  displays the triple $x_2y_2|y_1$. Furthermore, $\lca_T(x_1,x_2)$ has
  distinct children $v_1$ and $v_2$ such that $x_1,y_1\prec_T v_1$ and
  $x_2,y_2\prec_T v_1$.  Now since $\sigma(y_1)=\sigma(y_2)=\sigma(y_3)$,
  the two arcs $(x_1,y_3)$ and $(x_2,y_3)$ imply that
  $\lca_T(x_1,y_3)\preceq_T\lca_T(x_1,y_1)\preceq_T v_1$ and
  $\lca_T(x_2,y_3)\preceq_T\lca_T(x_2,y_2)\preceq_T v_2$, respectively.
  Since $v_1$ and $v_2$ are incomparable w.r.t.\ $\preceq_{T}$, this is a
  contradiction.
\end{proof}

\begin{lemma}
  \label{lem:no-Fx-implies-Nx}
  Let $(\G,\sigma)$ be a properly 2-colored graph. Then $(\G,\sigma)$
  satisfies \AX{(N1)} if it does not contain an induced F1-graph, it
  satisfies \AX{(N2)} if it does not contain an induced F2-graph, and
  it satisfies \AX{(N3)} if is contains neither an induced F1-graph nor
  an induced F3-graph.
\end{lemma}
\begin{proof}
We employ contraposition and thus show that $(\G=(V,E),\sigma)$ contains a
forbidden subgraph whenever \AX{(N1)}, \AX{(N2)} or \AX{(N3)} are violated.

Assume that \AX{(N1)} is not satisfied. Thus, there are two
$\rthin$-classes $\alpha$ and $\beta$ with
$\alpha\cap N(\beta)=\beta\cap N(\alpha)=\emptyset$ for which
$N(\alpha)\cap N(N(\beta))\neq\emptyset$ or
$N(\beta)\cap N(N(\alpha))\neq\emptyset$. We can w.l.o.g.\ assume that
$N(\beta)\cap N(N(\alpha))\neq\emptyset$. Note that
$\alpha\cap N(\beta)=\emptyset$ implies that $(y,x)\notin E$ for all
$x\in\alpha, y\in\beta$. Likewise $(x,y)\notin E$ for all
$x\in\alpha, y\in\beta$, since $\beta\cap N(\alpha)=\emptyset$. Let
$x_1\in\alpha$, $y_2\in\beta$ and
$x_2\in N(\beta)\cap N(N(\alpha))\neq\emptyset$. It must hold
$(x_1,y_2),(y_2,x_1)\notin E$ by the arguments above. Since
$x_2\in N(\beta)$, we have $(y_2,x_2)\in E$. Moreover,
$\sigma(x_1)=\sigma(x_2)\ne\sigma(y_2)$, since $(\G,\sigma)$ is properly
colored. Clearly, $x_2\in N(N(\alpha))$ implies that
$N(\alpha)\ne\emptyset$. Now, let $y_1\in N(\alpha)$ be a vertex such that
$(y_1,x_2)\in E$, which must exist as a consequence of
$x_2\in N(N(\alpha))$. We have $(x_1,y_1)$ since $y_1\in N(\alpha)$ and
thus $\sigma(y_1)=\sigma(y_2)\ne\sigma(x_1)=\sigma(x_2)$. Finally,
$(y_1,x_2)\in E$ immediately implies that $y_1\ne y_2$. In summary,
$(x_1,y_1),(y_1,x_2),(y_2,x_2)\in E$ and $(x_1,y_2),(y_2,x_1)\notin E$, and
thus $(\G,\sigma)$ contains an induced F1-graph.

Now assume that \AX{(N2)} is not satisfied and thus,
$N(N(N(\alpha)))\not\subseteq N(\alpha)$ for some $\rthin$-class $\alpha$.
Note, the latter implies that $N(N(N(\alpha)))\neq \emptyset$. Hence, there
is a vertex $y_2\in N(N(N(\alpha)))$ such that $y_2\notin N(\alpha)$. Thus,
there is a vertex $x_1\in \alpha$ such that $(x_1,y_2)\notin E$.  By the
definition of neighborhoods and since $y_2\in N(N(N(\alpha)))$, we find
vertices $y_1\in N(\alpha)$ and $x_2\in N(N(\alpha))$ such that
$(x_1,y_1),(y_1,x_2),(x_2,y_2)$. Since $(\G,\sigma)$ is properly colored,
we must have $\sigma(x_1)=\sigma(x_2)\ne\sigma(y_1)=\sigma(y_2)$. Moreover,
$(x_1,y_2)\notin E$ together with $(x_2,y_2)\in E$ and $(x_1,y_1)\in E$
implies $x_1\ne x_2$ and $y_1\ne y_2$, respectively. We conclude that the
subgraph induced by $x_1,x_2,y_1,y_2$ contains an induced F2-graph.

Finally, assume that \AX{(N3)} is not satisfied.  Hence, there are two
$\rthin$-classes $\alpha$ and $\beta$ with
$\alpha\cap N(N(\beta))=\beta\cap N(N(\alpha))=\emptyset$ and
$N(\alpha)\cap N(\beta)\neq \emptyset$, but (i)
$N^-(\alpha)\ne N^-(\beta)$, or (ii) neither $N(\alpha)\subseteq N(\beta)$
nor $N(\beta)\subseteq N(\alpha)$. Note,
$N(\alpha)\cap N(\beta)\neq \emptyset$ implies that there a vertices
$x_1\in\alpha$ and $x_2\in\beta$ with $\sigma(x_1)=\sigma(x_2)$ since
$(\G,\sigma)$ is properly 2-colored. In particular, there must be a vertex
$y_3$ with $(x_1,y_3),(x_2,y_3)\in E$ and thus
$\sigma(x_1)=\sigma(x_2)\ne\sigma(y_3)$.

Now consider Case~(i) and suppose that $N^-(\alpha)\ne N^-(\beta)$.  Thus
we can assume w.l.o.g.\ that there is a $y^*$ with $(y^*,x_2)\in E$ but
$(y^*,x_1)\notin E$.  Note, $(x_1,y^*)\notin E$, since otherwise
$(x_1,y^*),(y^*,x_2)\in E$ would contradict
$\beta\cap N(N(\alpha))=\emptyset$.  Thus, $y^*\ne y_3$ since
$(x_1,y^*)\notin E$ and $(x_1,y_3)\in E$.  Furthermore,
$\sigma(y^*)=\sigma(y_3)\ne\sigma(x_1)=\sigma(x_2)$, since $(\G,\sigma)$ is
properly 2-colored.  In summary, $(y^*,x_2),(x_1,y_3),(x_2,y_3)\in E$ and
$(y^*,x_1),(x_1,y^*)\notin E$ which implies that $(\G,\sigma)$ contains
an induced F1-graph.

Now consider Case~(ii) and assume that it holds neither
$N(\alpha)\subseteq N(\beta)$ nor $N(\beta)\subseteq N(\alpha)$.  Clearly,
the latter implies $N(\alpha)\neq \emptyset$ and $N(\beta)\neq
\emptyset$. The latter two arguments imply that there must be two distinct
vertices $y_1\in N(\alpha)\setminus N(\beta)$ and
$y_2\in N(\beta)\setminus N(\alpha)$ and, therefore,
$(x_1,y_1),(x_2,y_2)\in E$ and $(x_1,y_2),(x_2,y_1)\notin E$.  It follows
that $y_1\ne y_3$ and $y_2\ne y_3$ and
$\sigma(y_1)=\sigma(y_2)=\sigma(y_3)\ne\sigma(x_1)=\sigma(x_2)$.  This and
$(x_1,y_1),(x_2,y_2),(x_1,y_3),(x_2,y_3)\in E$ together with
$(x_1,y_2),(x_2,y_1)\notin E$ implies that $(\G,\sigma)$ contains an
induced F3-graph.
\end{proof}

Based on the latter findings we obtain here a new characterization of
2-colored BMGs that is not restricted to connected graphs.
\begin{theorem}
  \label{thm:newCharacterizatio}
  A properly 2-colored graph is a BMG if and only if it is sink-free and
  does not contain an induced F1-, F2-, or F3-graph.
\end{theorem}
\begin{proof}
  Suppose that $(\G,\sigma)$ is 2-colored BMG and $\mathfrak{C}$ be the set
  of its connected components. By Lemma \ref{lem:forbidden-subgraphs},
  $(\G,\sigma)$ does not contain an induced F1-, F2- or F3-graph.
  Moreover, by Lemma~\ref{lem:BMGunion}, $(\G[C],\sigma_{|C})$ must be a
  2-colored BMG for all $C\in \mathfrak{C}$.  Hence, we can apply
  Thm.~\ref{thm:charact-2BMG} to conclude that each $(\G[C],\sigma_{|C})$
  satisfies \AX{(N0)}-\AX{(N3)}. Since every $x\in V$ is contained in some
  $\rthin$-class, \AX{(N0)} is equivalent to $N(x)\ne\emptyset$, i.e.,
  $(\G,\sigma)$ is sink-free.

  Now suppose that $(\G,\sigma)$ is properly 2-colored and sink-free,
	and that it does not contain an induced F1-, F2- and F3-graph. By
  Lemma \ref{lem:no-Fx-implies-Nx}, $(\G,\sigma)$ satisfies
  \AX{(N1)}-\AX{(N3)}. Thus, in particular, each connected component of
  $(\G,\sigma)$ is sink-free and satisfies and \AX{(N1)}-\AX{(N3)}.  Note,
  $N(x)\ne\emptyset$ implies that the connected components of $(\G,\sigma)$
  contain at least one arc and, by assumption, they are properly 2-colored.
  Moreover, this implies that \AX{(N0)} is satisfied for every connected
  component of $(\G,\sigma)$.  Hence, Thm.~\ref{thm:charact-2BMG} implies
  that every connected component of $(\G,\sigma)$ is a 2-colored BMG. By
  Lemma \ref{lem:BMGunion}, $(\G,\sigma)$ is also a 2-colored BMG.
\end{proof}

\section{Complexity of 2-BMG modification problems}
\label{sect:NPcomplete2}

In real-live applications, we have to expect that graphs estimated from
empirical best match data will contain errors. Therefore, we consider the
problem of correcting erroneous and/or missing arcs. Formally, we consider
the following graph modification problems for properly colored digraphs.
\begin{problem}[\PROBLEM{$\ell$-BMG Deletion}]\ \\
  \begin{tabular}{ll}
    \emph{Input:}    & A properly $\ell$-colored digraph $(\G =(V,E),\sigma)$
                       and an integer $k$.\\
    \emph{Question:} & Is there a subset $F\subseteq E$ such that $|F|\leq
                       k$ and $(\G- F,\sigma)$ is an $\ell$-BMG?
  \end{tabular}
\end{problem}
It is worth noting that \PROBLEM{$\ell$-BMG Deletion} does not always have
a feasible solution. In particular, if $(\G,\sigma)$ contains a sink, no
solution exits for any $\ell>1$ as a consequence of
  Thm.~\ref{thm:BMG-charac-via-R-F} and the fact that we only delete
  arcs.  In contrast, it is always possible to obtain a BMG from a
properly colored digraph $(\G,\sigma)$ if arc insertions are allowed.  To
see this, observe that the graph $(\G',\sigma)$ with $V(\G')=V(\G)$ that
contains all arcs between vertices of different colors is a BMG, since it
is explained the tree with leaf set $V(\G')$ in which all leaves are
directly attached to the root.  This suggests that the following two
problems are more relevant for practical applications:
\begin{problem}[\PROBLEM{$\ell$-BMG Editing}]\ \\
  \begin{tabular}{ll}
    \emph{Input:}    & A properly $\ell$-colored digraph $(\G =(V,E),\sigma)$
                       and an integer $k$.\\
    \emph{Question:} & Is there a subset $F\subseteq V\times V \setminus
                       \{(v,v)\mid v\in V\}$ such
                       that $|F|\leq k$\\
                     & and $(\G\symdiff F,\sigma)$
                       is an $\ell$-BMG?
  \end{tabular}
\end{problem}

\begin{problem}[\PROBLEM{$\ell$-BMG Completion}]\ \\
  \begin{tabular}{ll}
    \emph{Input:}    & A properly $\ell$-colored digraph $(\G =(V,E),\sigma)$
                       and an integer $k$.\\
    \emph{Question:} & Is there a subset $F\subseteq V\times V \setminus
                       (\{(v,v)\mid v\in V\} \cup E)$ such
                       that $|F|\leq k$\\
                     & and $(\G+ F,\sigma)$
                       is an $\ell$-BMG?
  \end{tabular}
\end{problem}

In this section, we consider decision problems related to modifying
$2$-colored digraphs. The general case with an arbitrarily large number
$\ell\ge2$ of colors will be the subject of the following section. For
$\ell=2$, we will show that both \PROBLEM{$2$-BMG Deletion} and
\PROBLEM{$2$-BMG Editing} are NP-complete by reduction from the
\PROBLEM{Exact 3-Cover} problem (\PROBLEM{X3C}), one of Karp's famous 21
NP-complete problems \cite{Karp1972}.
\begin{problem}[\PROBLEM{Exact 3-Cover (X3C)}]\ \\
  \begin{tabular}{ll}
    \emph{Input:}    & A set $\mathfrak{S}$ with $|\mathfrak{S}|=3t$
                       elements and a collection $\mathcal{C}$ of
                       3-element subsets of $\mathfrak{S}$.\\
    \emph{Question:} & Does $\mathcal{C}$ contain an exact cover for
                       $\mathfrak{S}$, i.e., a subcollection
                       $\mathcal{C}'\subseteq\mathcal{C}$
                       such that \\
                     & every element of $\mathfrak{S}$ occurs in exactly
                       one member of $\mathcal{C}'$?
  \end{tabular}
\end{problem}
\noindent
An exact 3-cover $\mathcal{C}'$ of $\mathfrak{S}$ with 
$|\mathfrak{S}|=3t$ is
necessarily of size $|\mathcal{C}'|=t$ and satisfies
$\bigcup_{C\in\mathcal{C}'}C=\mathfrak{S}$.
\begin{theorem}{\cite{Karp1972}}
  \PROBLEM{X3C} is NP-complete.
  \label{thm:NPc-X3c}
\end{theorem}

In the following, we will make extensive use of properly 2-colored
  digraphs that contain all possible arcs:
\begin{definition}
  A \emph{bi-clique} of a colored digraph $(\G,\sigma)$ is a subset of
  vertices $C\subseteq V(\G)$ such that (i) $|\sigma(C)|=2$ and (ii)
  $(x,y)\in E(\G[C])$ if and only if $\sigma(x)\ne\sigma(y)$ for all
  $x,y\in C$.  A colored digraph $(\G,\sigma)$ is a \emph{bi-cluster graph}
  if all its connected components are bi-cliques.
\end{definition}

\begin{figure}[t]
  \begin{center}
    \includegraphics[width=0.18\linewidth]{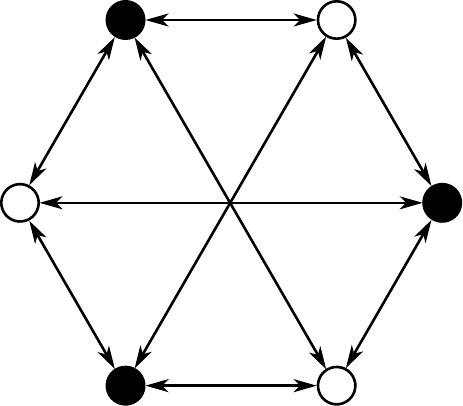}
  \end{center}
  \caption[]{A (sub)graph induced by a bi-clique consisting of 3 $\Black$ and
    3 $\White$ vertices. It has 18 arcs in total.}
  \label{fig:colored_clique_3_3}
\end{figure}
In a bi-clique, all arcs between vertices of different color are present.
Thus a bi-clique with $n$ and $m$ vertices in the two color classes has
$2nm$ arcs, see Fig.~\ref{fig:colored_clique_3_3} for the case $n=m=3$.  We
emphasize that, in contrast to the definition used in
\cite{Hellmuth:2020a}, single vertex graphs are not considered as
bi-clique.

\begin{figure}[t]
  \begin{center}
    \includegraphics[width=0.75\linewidth]{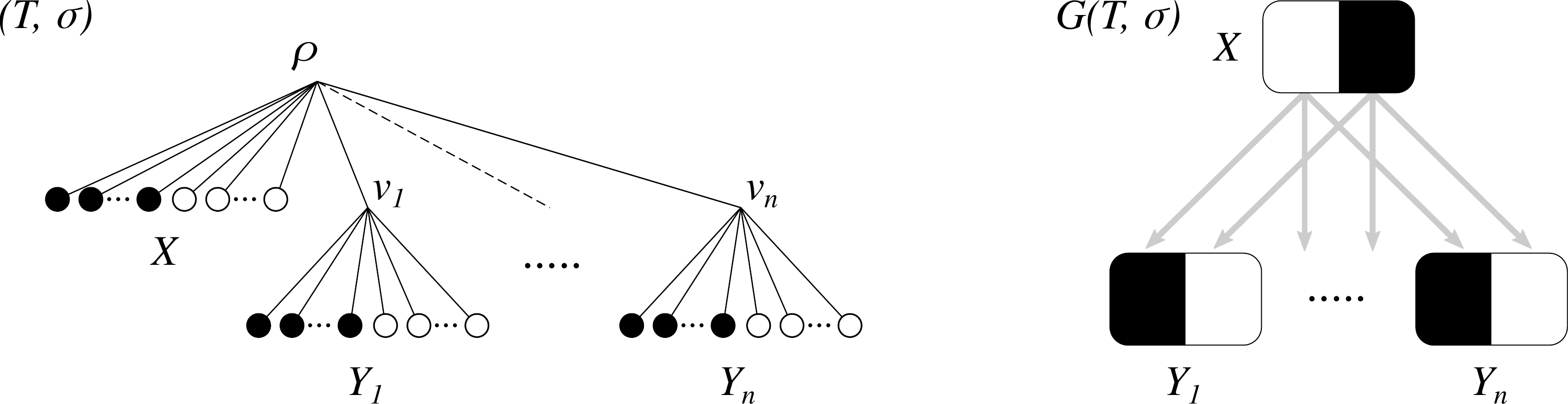}
  \end{center}
  \caption[]{A tree $(T,\sigma)$ whose BMG $\G(T,\sigma)$
  		contains bi-cliques $X$ and $Y_1,\dots,Y_n$.
  	The thick
    gray arrows indicate that all arcs in that direction exist between the
    respective sets.}
  \label{fig:component_tree}
\end{figure}

We start with a simple construction of a subclass of BMGs from disconnected
2-colored bi-cluster graph:
\begin{lemma}\label{lem:BMGspecial}
  Let $(\G,\sigma)$ be a 2-colored bi-cluster graph with at least two
  connected components, let $\mathfrak{C}$ be the set of connected
  components of $(\G,\sigma)$, and fix one of these
  connected components $X\in \mathfrak{C}$. Now, let $(\G',\sigma)$ be the 
  graph obtained
  from $(\G,\sigma)$ by adding all arcs $(x,y)$ with
  $x\in X$ and
  $y\in\bigcup_{Y\in \mathfrak{C}\setminus \{X\}} Y$ for which
  $\sigma(x)\neq\sigma(y)$. Then $(\G',\sigma)$ is a BMG.
\end{lemma}
\begin{proof}
  To see that $(\G',\sigma)$ is a BMG it suffices to show that there is a
  tree $(T,\sigma)$ that explains $(\G',\sigma)$. To this end, consider the
  tree $(T,\sigma)$ as shown in Fig.~\ref{fig:component_tree} and its
  BMG $\G(T,\sigma)$.  Observe first that, for all $x,y\in X$, it holds
  $\lca_T(x,y)=\rho = \lca_T(x,y')=\lca(x',y)$ for all $x',y'\in
  L(T)$. Hence, $X$ is a bi-clique and there are arcs from all vertices in
  $X$ to all vertices of distinct color in
  $Y_i \in \mathfrak{C}\setminus \{X\}$.  Moreover, for all
  $x,y\in Y_i \in \mathfrak{C}\setminus \{X\}$ it holds that
  $\lca_T(x,y)=v_i \preceq_T \lca_T(x,y')=\lca(x',y)$ for all
  $x',y'\in L(T)$.  Hence, $Y_i$ is a bi-clique for all
  $Y_i \in \mathfrak{C}\setminus \{X\}$.  Finally, for all
  $x,y\in Y_i \in \mathfrak{C}\setminus \{X\}$ and all
  $x',y'\in L(T)\setminus Y_i$ it holds
  $\lca_T(x,y)=v_i \prec_T \lca_T(x,y')=\lca(x',y) = \rho$ which implies
  that there are no arcs from vertices in $Y_i$ to vertices in $X$ and no
  arcs between distinct $Y_i, Y_j \in \mathfrak{C}\setminus \{X\}$.  In
  summary, $\G(T,\sigma) = (\G',\sigma)$ and hence, $(\G',\sigma)$ is a
  BMG.
\end{proof}

We are now in the position to prove NP-completeness of \PROBLEM{$2$-BMG
  Editing}. The strategy of the NP-hardness proof is very similar to the
one used in \cite{El-Mallah:1988} and \cite{Liu:2012} to show the
  NP-hardness of \PROBLEM{Cograph Editing}. Nevertheless, although
  similar in fashion, our proof has subtle but important differences when
  compared to the proofs provided in \cite{El-Mallah:1988} and
  \cite{Liu:2012}.  In particular, at the heart of our construction are
  2-colored bi-cliques rather than complete graphs.

\begin{theorem}
  \label{thm:2col-NP}
  \PROBLEM{$2$-BMG Editing} is NP-complete.
\end{theorem}
\begin{proof}
  Since BMGs can be recognized in polynomial time by
    Cor.~\ref{cor:bmg-rec-polytime}, the \PROBLEM{$2$-BMG Editing} problem
  is clearly contained in NP.  To show the NP-hardness, we use reduction
  from \PROBLEM{X3C}.

  Let $\mathfrak{S}$ with $|\mathfrak{S}|=n=3t$ and
  $\mathcal{C}=\{C_1,\dots,C_m\}$ be an instance of \PROBLEM{X3C}. Clearly,
  if $m=t$ the \PROBLEM{X3C} problem becomes trivial and thus, we assume
  w.l.o.g.\ that $m>t$. The latter implies that every solution
  $\mathcal{C}'$ of \PROBLEM{X3C} satisfies
  $\mathcal{C}'\subsetneq\mathcal{C}$. Moreover, we assume w.l.o.g.\ that
  $C_i\ne C_j$, $1\leq i<j\leq m$.  We construct an instance
  $(\G=(V,E),\sigma,k)$, where $(\G,\sigma)$ is colored with the two colors
  $\Black$ and $\White$, of the \PROBLEM{$2$-BMG Editing} problem as
  follows: First, we construct a bi-clique $S$ consisting of a $\Black$
  vertex $s_{b}$ and a $\White$ vertex $s_{w}$ for every
  $s\in \mathfrak{S}$.  Thus the subgraph induced by $S$ has $6t$ vertices
  and $r\coloneqq 18t^2$ arcs in total.  Let
  $q\coloneqq 3\times[6r(m-t)+r-18t]$.  For each of the $m$ subsets $C_i$
  in $\mathcal{C}$, we introduce two bi-cliques $X_i$ and $Y_i$, where
  $X_i$ consists of $r$ $\Black$ and $r$ $\White$ new vertices, and $Y_i$
  consists of $q$ $\Black$ and $q$ $\White$ new vertices.  In addition to
  the arcs provided by bi-cliques constructed in this manner, we add the
  following additional arcs:
  \begin{description}[noitemsep,nolistsep]
  \item[--] $(x,y)$ for every $x\in X_i$ and $y\in Y_i$ with
    $\sigma(x)\ne\sigma(y)$ (note $(y,x)\notin E$),
  \item[--] $(x,s_{b})$ for every $\White$ vertex $x\in X_i$ and
    every element $s\in C_i$, and,
  \item[--] $(x,s_{w})$ for every $\Black$ vertex $x\in X_i$ and every
    element $s\in C_i$.
  \end{description}
  This construction is illustrated in Fig.~\ref{fig:reduction}. Clearly,
  $(\G,\sigma)$ is properly colored, and the reduction can be computed in
  polynomial time.

  \begin{figure}[ht]
    \begin{center}
      \includegraphics[width=0.65\linewidth]{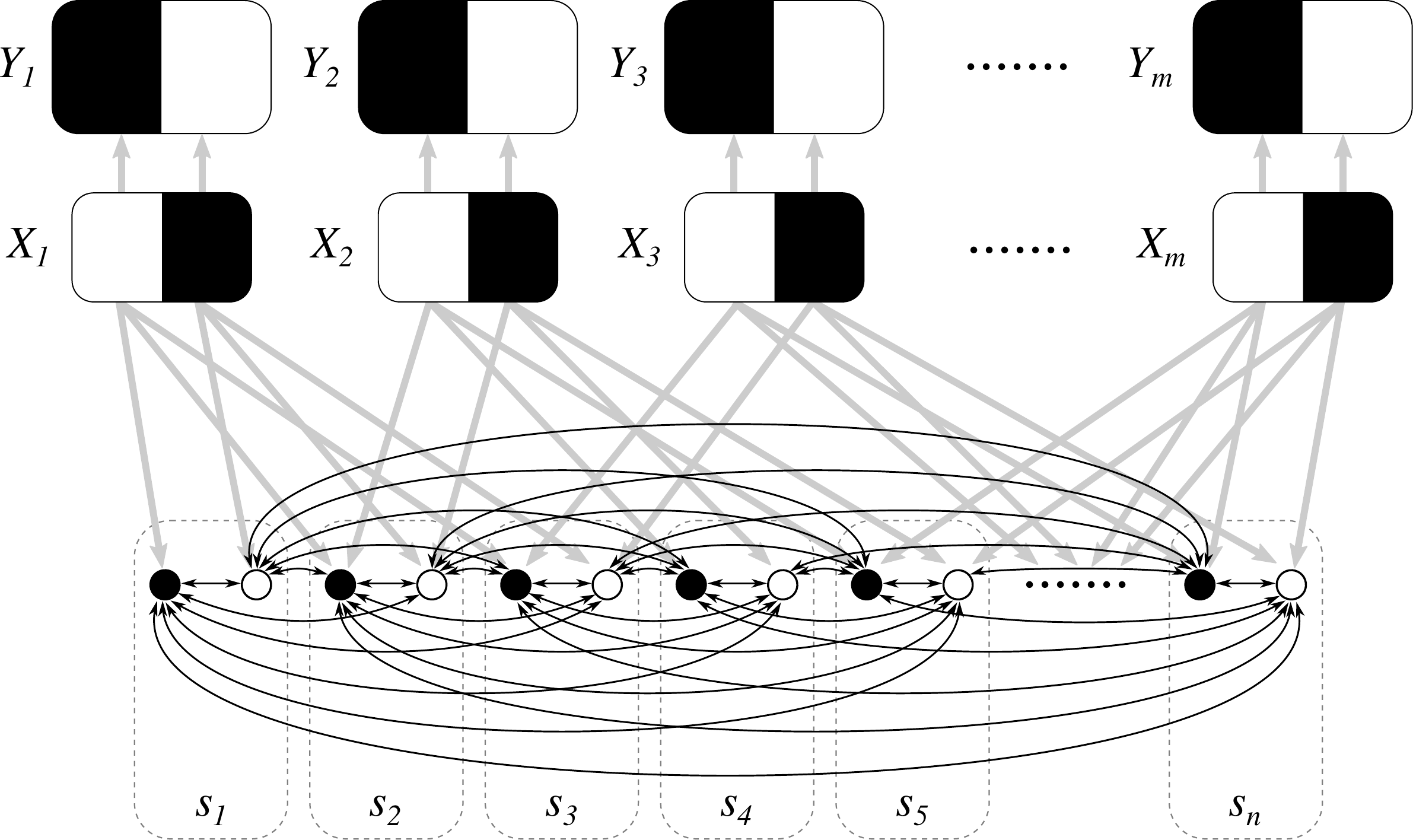}
    \end{center}
    \caption[]{Illustration of the reduction from \PROBLEM{Exact 3-Cover}.
      The thick gray arrows indicate that all arcs from that set to another
      set/vertex exist.  The illustration emphasizes the analogy to
      \cite{El-Mallah:1988} and \cite{Liu:2012}.}
    \label{fig:reduction}
  \end{figure}

  We set $k\coloneqq 6r(m-t)+r-18t$ and show that there is a $t$-element
  subset $\mathcal{C}'$ of $\mathcal{C}$ that is a solution of
  \PROBLEM{X3C} if and only \PROBLEM{$2$-BMG Editing} with input
  $(\G,\sigma,k)$ has a yes-answer.  We emphasize that the coloring
  $\sigma$ remains unchanged in the proof below.

  \smallskip

  First suppose that \PROBLEM{X3C} with input $\mathfrak{S}$ and
  $\mathcal{C}$ has a yes-answer.  Thus, there is a $t$-element subset
  $\mathcal{C}'$ of $\mathcal{C}$ such that
  $\bigcup_{C\in\mathcal{C}'}C=\mathfrak{S}$.  We construct a set $F$ and
  add, for all $C_i\in\mathcal{C}\setminus\mathcal{C}'$ and all
    $s\in C_i$, the arcs $(x,s_{w})$ for every $\Black$ vertex $x\in X_i$
    and the arcs $(x,s_{b})$ for every $\White$ vertex $x\in X_i$.
  Since $|C_i|=3$ for every $C_i\in \mathcal{C}$ and
  $|\mathcal{C}\setminus\mathcal{C}'|=m-t$, the set $F$ contains exactly
  $6r(m-t)$ arcs, so far. Now, we add to $F$ all arcs $(s_b,s'_w)$ and
  $(s_w,s'_b)$ whenever the corresponding elements $s$ and $s'$ belong to
  distinct elements in $\mathcal{C}'$, i.e., there is no $C\in\mathcal{C}'$
  with $\{s, s'\}\subset C$.  Therefore, the subgraph of $\G-F$ induced by
  $\mathfrak{S}$ is the disjoint union of $t$ bi-cliques, each consisting 
  of exactly $3$ $\Black$ vertices, $3$ $\White$ vertices, and $18$ arcs.
  Hence, $F$ contains, in addition to the $6r(m-t)$ arcs, further $r-18t$
  arcs. Thus $|F|=k$. This completes the construction of $F$.

  Since $F$ contains only arcs but no non-arcs of $\G$, we have
  $\G\symdiff F = \G-F$.  It remains to show that $\G\symdiff F$ is a
  BMG. To this end observe that $\G\symdiff F$ has precisely $m$ connected
  components that are either induced by $X_i\cup Y_i$ (in case
  $C_i\in\mathcal{C}\setminus\mathcal{C}'$ ) or $X_i\cup Y_i\cup S'$ where
  $S'$ is a bi-clique containing the six vertices corresponding to the
  elements in $C_i\in\mathcal{C}'$. In particular, each of these components
  corresponds to the subgraph as specified in Lemma~\ref{lem:BMGspecial}.
  To see this, consider the bi-cluster graph consisting of the
    subgraphs induced by the bi-cliques $X_i$ and $Y_i$, and additionally
    $S'$ in the second case. Now fix the connected component $X_i$ of this
    graph, and add all arcs from the vertices in $X_i$ to
    differently-colored vertices in the remaining component(s).  Thus,
  we can apply Lemma~\ref{lem:BMGspecial} to conclude that every
    connected component is a BMG.  In particular, all of these subgraphs
  contain at least one $\Black$ and one $\White$ vertex. Hence,
  Lemma~\ref{lem:BMGunion} implies that $(\G\symdiff F,\sigma)$ is a BMG.

  \smallskip

  Now, suppose that \PROBLEM{2-BMG Editing} with input $(\G,\sigma)$ has a
  yes-answer.  Thus, there is a set $F$ with $|F|\le k$ such that
  $(\G\symdiff F,\sigma)$ is a BMG.  We will prove that we have to delete
  an arc set similar to the one as constructed above.  First note that the
  number of vertices affected by $F$, i.e. vertices incident to
  inserted/deleted arcs, is at most $2k$.  Since
  $2k<q=|\{y\in Y_i\mid \sigma(y)=\Black\}|=|\{y\in Y_i\mid
  \sigma(y)=\White\}|$ for every $1\le i\le m$, we have at least on
  $\Black$ vertex $b_i\in Y_i$ and at least one $\White$ vertex
  $w_i\in Y_i$ that are unaffected by $F$.  Recall that $S$ is the
    bi-clique that we have constructed from a $\Black$ vertex $s_{b}$ and a
    $\White$ vertex $s_{w}$ for every $s\in \mathfrak{S}$.  We continue 
    by
  proving
  \par\noindent
  \begin{claim}
    \label{clm:at-most-one-Xi}
    Every vertex $s\in S$ has in-arcs from at most one $X_i$ in
    $\G\symdiff F$.
  \end{claim}
  \begin{claim-proof}\item \emph{Proof:}
    Assume w.l.o.g.\ that $s$ is $\Black$ and, for contradiction, that there
    are two distinct vertices $x_1\in X_i$ and $x_2\in X_j$ with $i\ne j$
    and $(x_1,s), (x_2,s)\in E\symdiff F$.  Clearly, both $x_1$ and $x_2$
    are $\White$.  As argued above, there are two (distinct) $\Black$ vertices
    $b_1\in Y_i$ and $b_2\in Y_j$ that are not affected by $F$.  Thus,
    $(x_1,b_1)$ and $(x_2,b_2)$ remain arcs in $\G\symdiff F$, whereas
    $(x_1,b_2)$ and $(x_2,b_1)$ are not arcs in $\G\symdiff F$, since they
    do not form arcs in $\G$.  In summary, we have five distinct vertices
    $x_1,x_2,b_1,b_2,s$ with
    $\sigma(x_1)=\sigma(x_2)\ne\sigma(b_1)=\sigma(b_2)=\sigma(s)$, arcs
    $(x_1,b_1),(x_2,b_2),(x_1,s),(x_2,s)$ and non-arcs
    $(x_1,b_2),(x_2,b_1)$.  Thus $(\G\symdiff F,\sigma)$ contains an
    induced F3-graph.  By Lemma~\ref{lem:forbidden-subgraphs},
    $(\G\symdiff F,\sigma)$ is not a BMG; a contradiction. \hfill$\diamond$
  \end{claim-proof}

  By Claim~\ref{clm:at-most-one-Xi}, every vertex in $S$ has in-arcs from
  at most one $X_i$.  Note each $X_i$ has $r$ $\Black$ and $r$ $\White$
  vertices.  Since each element in $S$ is either $\White$ or $\Black$, each
  single element in $S$ has at most $r$ in-arcs.  Since $|S| = 2n$, we
  obtain at most $2rn = 2r(3t)=6rt$ such arcs in $\G\symdiff F$.  In
  $\G$, there are in total $6rm$ arcs from the vertices in all $X_i$ to the
  vertices in $S$.  By Claim \ref{clm:at-most-one-Xi}, $F$ contains at
  least $6r(m-t)$ deletions.  It remains to specify the other at most
  $r-18t$ arc modifications.  To this end, we show first
  \begin{claim}
    \label{clm:at-presisely-Xi}
    Every vertex $s\in S$ has in-arcs from precisely one $X_i$ in
    $\G\symdiff F$.
  \end{claim}
  \begin{claim-proof}\item \emph{Proof:}
    Assume that there is a vertex $s\in S$ that has no in-arc from any
    $X_i$.  Hence, to the aforementioned $6r(m-t)$ deletions we must add
    $r$ further deletions.  However, at most $r-18t$ further edits are
    allowed; a contradiction.  \hfill$\diamond$
  \end{claim-proof}
  So far, $F$ contains only arc-deletions.  For the  next arguments, we
  need the following two statements:
  \begin{claim}
    \label{clm:no-insertion-XiXj}
    The modification set $F$ does not insert any arcs between $X_i$ and
    $X_j$ with $i\ne j$.
  \end{claim}
  \begin{claim-proof}\item\emph{Proof:}
    Assume for contradiction that $F$, and thus $\G\symdiff F$,
    contains an arc $(x_1,x_2)$ with $x_1\in X_i$, $x_2\in X_j$ and
    $i\ne j$.  W.l.o.g.\ assume that $x_1$ is $\White$ and $x_2$ is
    $\Black$.  As argued above there are $\Black$, resp., $\White$ vertices
    $b, w\in Y_j$ that are unaffected by $F$.  Therefore, $(x_2,w)$ and
    $(b,w)$ remain arcs in $\G\symdiff F$, whereas $(x_1,b)$ and $(b,x_1)$
    are not arcs in $\G\symdiff F$ since they do not form arcs in $\G$. In
    summary, $(x_1,x_2),(b,w),(x_2,w)$ are arcs in $\G\symdiff F$ while
    $(x_1,b),(b,x_1)$ are not arcs in $\G\symdiff F$.  Since moreover
    $\sigma(x_1)=\sigma(w)\ne\sigma(b)=\sigma(x_2)$,
    $(\G\symdiff F,\sigma)$ contains an induced F1-graph.  By
    Lemma~\ref{lem:forbidden-subgraphs}, $(\G\symdiff F,\sigma)$ is not a
    BMG; a contradiction.  \hfill$\diamond$
  \end{claim-proof}
  \par\noindent
  \begin{claim}
    \label{clm:deletions-in-S}
    Let $s_1,s_2\in S$ be vertices with in-arcs $(x_1,s_1)$, resp.,
    $(x_2,s_2)$ in $\G\symdiff F$ for some $x_1\in X_i$ and $x_2 \in X_j$
    with $i\ne j$.  Then $(s_1,s_2)$ and $(s_2,s_1)$ cannot be arcs in
    $\G\symdiff F$.
  \end{claim}
  \begin{claim-proof}\item\emph{Proof:}
    Assume w.l.o.g. that $(s_1,s_2)$ is an arc in $\G\symdiff F$ and that
    $s_1$ is $\Black$.  It follows that $x_1$ and $s_2$ are $\White$ and
    $x_2$ is $\Black$.  By construction of $\G$ and by
    Claim~\ref{clm:no-insertion-XiXj}, we clearly have
    $(x_1,x_2),(x_2,x_1)\notin E\symdiff F$.  In summary, we have four
    distinct vertices $x_1,x_2,s_1,s_2$ with
    $\sigma(x_1)=\sigma(s_2)\ne\sigma(s_1)=\sigma(x_2)$, arcs
    $(x_1,s_1),(x_2,s_2),(s_1,s_2)$ and non-arcs $(x_1,x_2),(x_2,x_1)$ in
    $\G\symdiff F$.  Thus $(\G\symdiff F,\sigma)$ contains an induced
    F1-graph.  By Lemma~\ref{lem:forbidden-subgraphs},
    $(\G\symdiff F,\sigma)$ is not a BMG; a contradiction.
    \hfill$\diamond$
  \end{claim-proof}
  In summary, $\G\symdiff F$ has the following property: Every $s\in S$ has
  in-arcs from exactly one $X_i$, and there are no arcs between two
    distinct vertices $s_1$ and $s_2$ in $S$ that have in-arcs from two
  different sets $X_i$ and $X_j$,  respectively. Since $|C_i|=3$ for
  every $C_i\in\mathcal{C}$, $(\G\symdiff F)[S]$ contains connected
  components of size at most $6$, i.e., the $\Black$ and $\White$ vertex
  for each of the three elements in $C_i$. Hence, the maximum number of
  arcs in $(\G\symdiff F)[S]$ is obtained when each of its connected
  components contains exactly these $6$ vertices and they form a
  bi-clique. In this case, $(\G\symdiff F)[S]$ contains $18t$ arcs. We
  conclude that $F$ contains at least another $r-18t$ deletion arcs for
  $S$. Together with the at least $6r(m-t)$ deletions between the $X_i$ and
  the elements of $S$, we have at least $6r(m-t)+r-18t=k\geq |F|$
  arc-deletions in $F$.  Since $|F|\le k$ by assumption, we obtain $|F|=k$.

  As argued above, the subgraph induced by $S$ is a disjoint union of $t$
  bi-cliques of $3$ $\White$ and $3$ $\Black$ vertices each.  Since all
  vertices of such a bi-clique have in-arcs from the same $X_i$ and these
  in-arcs are also in $\G$, we readily obtain the desired partition
  $\mathcal{C}'\subset\mathcal{C}$ of $\mathfrak{S}$.  In other words, the
  $C_i$ corresponding to the $X_i$ having out-arcs to vertices in $S$ in
  the edited graph $\G\symdiff F$ induce an exact cover of $\mathfrak{S}$.
\end{proof}

The set $F$ constructed in the proof of Thm.~\ref{thm:2col-NP} contains
only arc-deletions. This immediately implies
\begin{corollary}
  \label{cor:2col-NP-deletion}
  \PROBLEM{$2$-BMG Deletion} is NP-complete.
\end{corollary}

In order to tackle the complexity of the \PROBLEM{$2$-BMG Completion}, we
follow a different approach and employ a reduction from the \PROBLEM{Chain
  Graph Completion} problem. To this end, we need some additional
notation. An undirected graph $U$ is bipartite if its vertex set can be
partitioned into two non-empty disjoint sets $P$ and $Q$ such that
$V(U)=P\cupdot Q$ and every edge has one endpoint in $P$ and the other
endpoint in $Q$.  We write $U=(P\cupdot Q,\undirected{E})$ to emphasize
that $\undirected{E}$ is a set of undirected edges and that $U$ is
bipartite. Furthermore, we write $N(x)$ also for the neighborhood of a
vertex $x$ in an undirected graph. Thus $U$ is bipartite if and only if
$x\in P$ implies $N(x)\subseteq Q$ and $x\in Q$ implies $N(x)\subseteq P$.
\begin{definition}(\cite[cf.][]{Natanzon:2001,Yannakakis:1981}) An
  undirected, bipartite graph $U=(P\cupdot Q,\undirected{E})$ is a
  \emph{chain graph} if there is an order $\lessdot$ on $P$ such that
  $u\lessdot v$ implies $N(u)\subseteq N(v)$.
\end{definition}

The \PROBLEM{Chain Graph Completion} problem consists of finding a
minimum-sized set of additional edges that converts an arbitrary
undirected, bipartite graph into a chain graph. More formally, its decision
version can be stated as follows:
\begin{problem}[\PROBLEM{Chain Graph Completion} (\PROBLEM{CGC})]\ \\
  \begin{tabular}{ll}
    \emph{Input:}    & An undirected, bipartite graph
                       $U=(P\cupdot Q,\undirected{E})$
                       and an integer $k$.\\
    \emph{Question:} & Is there a subset $\undirected{F}\subseteq \{\{p,q\}\mid
                       (p,q)\in P\times Q\} \setminus \undirected{E}$ such
                       that $|\undirected{F}|\leq k$ \\
                     & and $U'\coloneqq (P\cupdot Q,\undirected{E}\cup
                       \undirected{F})$ is a chain graph?
  \end{tabular}
\end{problem}
It is shown in \cite{Yannakakis:1981} that \PROBLEM{CGC} is NP-complete.
Following \cite{Yannakakis:1981}, we say that two edges $\{u,v\}$ and
$\{x,y\}$ in an undirected graph $U$ are \emph{independent} if $u,v,x,y$
are pairwise distinct and the subgraph $U[\{u,v,x,y\}]$ contains no
additional edges. We will need the following characterization of chain
graphs:
\begin{lemma}{\cite[Lemma~1]{Yannakakis:1981}}
  \label{lem:independent-edges}
  An undirected, bipartite graph $U=(P\cupdot Q,\undirected{E})$ is a chain
  graph if and only if it does not contain a pair of independent edges.
\end{lemma}

\begin{theorem}
  \label{thm:2col-completion-NP}
  \PROBLEM{$2$-BMG Completion} is NP-complete.
\end{theorem}
\begin{proof}
  Since BMGs can be recognized in polynomial time by
    Cor.~\ref{cor:bmg-rec-polytime}, \PROBLEM{$2$-BMG Completion} is
  clearly contained in NP.  To show NP-hardness, we use a reduction from
  \PROBLEM{CGC}.  Let $(U=(P\cupdot Q,\undirected{E}),k)$ be an instance of
  \PROBLEM{CGC} with vertex sets $P=\{p_1,\dots,p_{|P|}\}$ and
  $Q=\{q_1,\dots,q_{|Q|}\}$. To construct an instance $(\G=(V,E),\sigma,k)$
  of the \PROBLEM{$2$-BMG Completion} problem, we set
  $ V = P\cupdot Q \cupdot R \cupdot \{b\} \cupdot \{w\}$ where
  $R=\{r_1,\dots,r_{|Q|}\}$ is a copy of $Q$. The vertices are colored
  $\sigma(p_i)=\sigma(r_j)=\sigma(b)=\Black$ and
  $\sigma(q_i)=\sigma(w)=\White$. The arc set $E$ contains $(q_i,r_i)$ and
  $(r_i,q_i)$ for $1\le i \le |Q|$, $(p_i,w)$ for $1\le i\le |P|$, $(w,b)$
  and $(b,w)$, and $(p,q)$ for every $\{p,q\}\in\undirected{E}$. This
  construction is illustrated in
  Fig.~\ref{fig:completion-reduction}. Clearly, $(\G,\sigma)$ is properly
  colored, and the reduction can be computed in polynomial time.  Moreover,
  it is easy to verify that $(\G,\sigma)$ is sink-free by construction, and
  thus, any graph $(\G',\sigma)$ obtained from $(\G,\sigma)$ by adding arcs
  is also sink-free.  As above, we emphasize that the coloring $\sigma$
  remains unchanged in the completion process.

  \begin{figure}[ht]
    \begin{center}
      \includegraphics[width=0.65\linewidth]{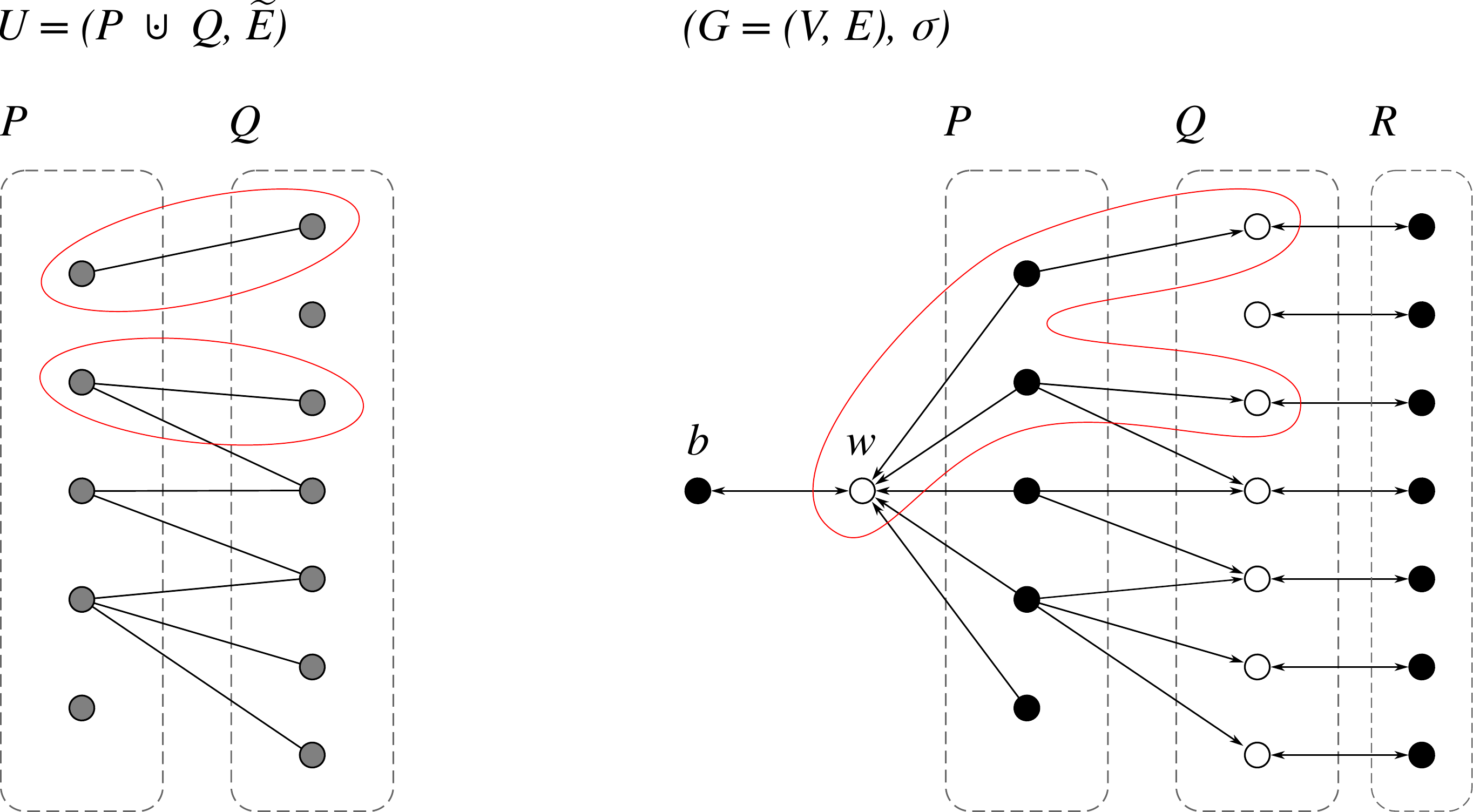}
    \end{center}
    \caption[]{Illustration of the reduction from \PROBLEM{CGC}. A pair of
      independent edges in $U$ and the corresponding induced F3-graph in
      $(\G,\sigma)$ are highlighted.}
    \label{fig:completion-reduction}
  \end{figure}

  A pair $(F,\undirected{F})$ with $F\subseteq P\times Q$ and an edge set
  $\undirected{F}=\{ \{p,q\} \mid (p,q)\in F\}$ will be called a
  \emph{completion pair} for the bipartite graph
  $U=(P\cupdot Q,\undirected{E})$ and the corresponding 2-colored digraph
  $(\G=(V,E),\sigma)$.
  \begin{claim}
    \label{clm:F-and-F-bar}
    If $(F,\undirected{F})$ is a completion pair, then
    $|F|=|\undirected{F}|$, $(p,q)\in F$ if and only if
    $\{p,q\}\in \undirected{F}$, and $(p,q)\in F\cup E$ if and only if
    $\{p,q\}\in \undirected{F}\cup\undirected{E}$.
  \end{claim}
  \begin{claim-proof}\item \emph{Proof:}
    First note that, by construction, $F$ contains only arcs from vertices
    in $P$ to vertices in $Q$. This together with the definition
    $\undirected{F}=\{ \{p,q\} \mid (p,q)\in F \}$ clearly implies
    $(p,q)\in F$ if and only if $\{p,q\}\in \undirected{F}$ and thus
    $|F|=|\undirected{F}|$.  By construction of our reduction we have
    $(p,q)\in E$ if and only if $\{p,q\}\in \undirected{E}$ and thus also
    $(p,q)\in E\cup F$ if and only if
    $\{p,q\}\in \undirected{E}\cup\undirected{F}$.  \hfill$\diamond$
  \end{claim-proof}

  Before we continue, observe that, for every pair of independent edges
  $\{p_1,q_1\},\{p_2,q_2\}\in\undirected{E}$, the subgraph of
    $(\G,\sigma)$ induced by $\{p_1,p_2,q_1,q_2,w\}$ is an F3-graph.
  Together with Lemmas~\ref{lem:forbidden-subgraphs}
  and~\ref{lem:independent-edges}, this implies that $(\G,\sigma)$ cannot
  be a BMG if $U$ is not a chain graph.  Eliminating these induced
  F3-graphs is closely connected to chain graph completion.  More precisely
  we will show:
  \begin{claim}
    \label{clm:bmg-implies-chain}
    Let $(F,\undirected{F})$ be a completion pair.  If $(\G+F,\sigma)$ is a
    BMG, then $U'=(P\cupdot Q,\undirected{E}\cup\undirected{F})$ is a chain
    graph.
  \end{claim}
  \begin{claim-proof}\item \emph{Proof:}
    Suppose that $(\G+F,\sigma)$ is a BMG and assume, for contradiction,
    that $U'=(P\cupdot Q,\undirected{E}\cup\undirected{F})$ is not a chain
    graph.  The latter and Lemma~\ref{lem:independent-edges} imply that
    $U'$ has two independent edges
    $\{p_1,q_1\},\{p_2,q_2\}\in\undirected{E}\cup\undirected{F}$.  Thus
    $\{p_1,q_2\},\{p_2,q_1\}\notin\undirected{E}\cup\undirected{F}$.  The
    latter arguments and Claim~\ref{clm:F-and-F-bar} imply that
    $(p_1,q_1),(p_2,q_2)\in E\cup F$ and
    $(p_1,q_2),(p_2,q_1)\notin E\cup F$.  Since moreover $(p_1,w),(p_2,w)$
    and $\sigma(p_1)=\sigma(p_2)\ne\sigma(q_1)=\sigma(q_2)=\sigma(w)$, it
    follows that the five distinct vertices $p_1,p_2,q_1,q_2,w$ induce an
    F3-graph in $(\G+F,\sigma)$.  By Lemma~\ref{lem:forbidden-subgraphs},
    $(\G+F,\sigma)$ cannot be a BMG; a contradiction.  \hfill$\diamond$
  \end{claim-proof}
  The converse is also true:
  \begin{claim}
    \label{clm:chain-implies-bmg}
    Let $(F,\undirected{F})$ be a completion pair for
    $U=(P\cupdot Q,\undirected{E})$, and suppose
    $U'=(P\cupdot Q,\undirected{E}\cup\undirected{F})$ is a chain graph.
    Then $(\G+F,\sigma)$ is a BMG.
  \end{claim}
  \begin{claim-proof}\item \emph{Proof:}
    By Thm.~\ref{thm:newCharacterizatio}, $(\G+F,\sigma)$ is a $2$-colored
    BMG if and and only if it is sink-free and does not contain an induced
    F1-, F2-, or F3-graph. Since $(\G,\sigma)$ is sink-free, this is also
    true for $(\G+F,\sigma)$. Thus it suffices to show that $(\G+F,\sigma)$
    does not contain an induced F1-, F2-, or F3-graph.

    Suppose that $(\G+F,\sigma)[u,u',v,v']$ is an induced F1-graph. Let $H$
    be a subgraph of $(\G+F,\sigma)[u,u',v,v']$ that is isomorphic to the
    \emph{essential} F1-graph, that is, the F1-graph as specified in
    Fig.~\ref{fig:forbidden-subgraphs} that contains only the solid-lined
    arcs and none of the dashed arcs while all other non-arcs remain
    non-arcs. In this case, there is an isomorphism $\varphi$ from $H$ to
    the essential F1-graph with vertex-labeling as in
    Fig.~\ref{fig:forbidden-subgraphs}.  Hence, $\varphi(u)$ corresponds to
    one of the vertices $x_1,x_2,y_1$ or $y_2$.  To simplify the
    presentation we will say that, in this case, ``$u$ plays the role of
    $\varphi(u)$ in an F1-graph''.

    The latter definition naturally extends to $F2$- and $F3$-graphs and we
    will use analogous language for $F2$- and $F3$-graphs.  Note, in the
    latter definition, it is not required that
    $\sigma(u) = \sigma(\varphi(u))$. Nevertheless, for
    $a,b\in \{u,u',v,v'\}$ with $\sigma(a)\neq \sigma(b)$ it always holds,
    by construction, that $\sigma(\varphi(a))\neq \sigma(\varphi(b))$.

    In the following, an in- or out-neighbor of a vertex is just called
    \emph{neighbor}.  A \emph{flank vertex} in an F1-, F2-, resp., F3-graph
    is a vertex that has only a single neighbor in the essential F1-, F2-,
    resp., F3-graph.  To be more precise, when referring to
    Fig.~\ref{fig:forbidden-subgraphs}, the flank vertices in an F1-graph
    and F2-graph are $x_1$ and $y_2$, while the flank vertices in an
    F3-graph are $y_1$ and $y_2$.

    Since $(F,\undirected{F})$ is a completion pair, by definition, $F$
    adds only arcs from $P$ to $Q$. Hence, each of the vertices in
    $R\cup\{b\}$ has a single neighbor in $(\G+F,\sigma)$ irrespective of
    the choice of $F$.  Therefore, if $u\in R\cup\{b\}$ is contained in an
    induced F1-, F2-, or F3-graph in $(\G,\sigma)$ or $(\G+F,\sigma)$, it
    must be a flank vertex. Observe first that $b$ can only play the role
    of $y_2$ in the F1- or F2-graph, since otherwise, the fact that $w$ is
    the single neighbor of $b$ in $(\G,\sigma)$ or $(\G+F,\sigma)$ implies
    that $w$ must play the role of $y_1$ in the F1- or F2-graph, which is
    not possible since $b$ is the single out-neighbor of $w$ and $F$ does
    not affect $w$. By similar arguments, none of the vertices in
    $R\cup\{b\}$ can play the role of $x_1$ in an F1- or F2-graph, or the
    role of $y_1$ or $y_2$ in an F3-graph in $(\G,\sigma)$ or
    $(\G+F,\sigma)$.  The vertex $w$ has only in-arcs from the elements in
    $P$ and from $b$.  Likewise, the vertices $q_i\in Q$ have only in-arcs
    from $P$ and from their corresponding vertex $r_i\in R$.  Therefore and
    since all elements in $P$ have only out-neighbors, it is an easy task
    to verify that none of the vertices in $R\cup\{b\}$ can play the role
    of $y_2$ in an F1- or F2-graph.  Thus none of the vertices in
    $R\cup\{b\}$ is part of an induced F1-, F2-, or F3-graph.

    Thus it suffices to investigate the subgraph $(\G',\sigma)$ of
    $(\G+F,\sigma)$ induced by $\{w\}\cup P \cup Q$ for the presence of
    induced F1-, F2-, and F3-graphs.  In $\G'$, none of the vertices in
    $\{w\}\cup Q$ have out-neighbors since $F\subseteq P\times Q$ does not
    affect $w$ and does not contain arcs from $q_i\in Q$ to any other
    vertex.  Thus, none of the vertices in $\{w\}\cup Q$ can play the role
    of $x_1$, $y_1$ or $y_2$ in an F1-, the role of $x_1$, $y_1$ or $x_2$
    in an F2-graph, or the role of $x_1$ or $x_2$ in an F3-graph. Since
    $\{w\}\cup Q$ has only in-arcs from $P$, and $P$ has no in-arcs in
    $\G'$, none of the vertices in $\{w\}\cup Q$ can play the role of $x_2$
    in an F1-graphs or the role of $y_2$ in an F2-graph.  Thus none of the
    vertices in $\{w\}\cup Q$ is part of an induced F1- or F2-graph. Hence,
    any induced F1- or F2-graph must be contained in $\G'[P]$. However, all
    vertices of $P$ are colored $\Black$, and hence $(\G'[P],\sigma_{|P})$
    cannot harbor an induced F1- or F2-graph.

    Suppose $(\G',\sigma)$ contains an induced F3-graph. Then there are
    five pairwise distinct vertices
    $x_1,x_2,y_1,y_2,y_3\in \{w\}\cup P \cup Q$ with coloring
    $\sigma(x_1)=\sigma(x_2)\ne\sigma(y_1)=\sigma(y_2)=\sigma(y_3)$
    satisfying $(x_1,y_1),(x_2,y_2),(x_1,y_3),(x_2,y_3)\in E\cup F$ and
    $(x_1,y_2),(x_2,y_1)\notin E\cup F$. Since $P$ has no in-arcs in
    $(\G',\sigma)$, it must hold that $y_1,y_2,y_3\notin P$.  Since
    $\sigma(\{w\}\cup Q)\neq \sigma(P)$ and $(\G',\sigma)$ is properly
    2-colored, we have $x_1,x_2\in P$. Since $w$ has in-arcs from all
    vertices in $P$ and $(x_1,y_2),(x_2,y_1)\notin E\cup F$, vertex $w$ can
    neither play the role of $y_1$ nor of $y_2$ in an F3-subgraph.  Thus,
    $y_1,y_2\in Q$.  Claim~\ref{clm:F-and-F-bar} therefore implies
    $\{x_1,y_1\},\{x_2,y_2\}\in \undirected{E}\cup\undirected{F}$ and
    $\{x_1,y_2\},\{x_2,y_1\}\notin \undirected{E}\cup\undirected{F}$.
    Hence, $U'$ contains a pair of independent edges.  By
    Lemma~\ref{lem:independent-edges}, it follows that $U'$ is not a chain
    graph; a contradiction.  \hfill$\diamond$
  \end{claim-proof}

  Together, Claims~\ref{clm:bmg-implies-chain}
  and~\ref{clm:chain-implies-bmg} imply that $(\G+F,\sigma)$ is a BMG if
  and only if $U'=(P\cupdot Q,\undirected{E}\cup\undirected{F})$ is a chain
  graph; see Fig.~\ref{fig:completion-tree} for an illustrative example.
  \begin{claim}
    \label{clm:minimum-arc-compl}
    If $F$ is a minimum-sized arc completion set such that $(\G+F,\sigma)$ is
    a BMG, then $F\subseteq P\times Q$.
  \end{claim}
  \begin{claim-proof}\item \emph{Proof:}
    Let $F$ be an arbitrary minimum-sized arc completion set, i.e.,
    $(\G+F,\sigma)$ is a BMG, and put $F'\coloneqq F\cap (P\times Q)$ and
    let $(F',\undirected{F'})$ be the corresponding completion pair.

    If $F'=F$, there is nothing to show.  Otherwise, we have $|F'|<|F|$ and
    minimality of $|F|$ implies that $(\G+F',\sigma)$ is not a BMG. By
    contraposition of Claim~\ref{clm:chain-implies-bmg}, we infer that
    $U'=(P\cupdot Q,\undirected{E}\cup\undirected{F'})$ is not a chain
    graph.  Hence, Lemma~\ref{lem:independent-edges} implies that $U'$
    contains a set of independent edges
    $\{p_1,q_1\},\{p_2,q_2\}\in \undirected{E}\cup\undirected{F'}$ and
    $\{p_1,q_2\},\{p_2,q_1\}\notin \undirected{E}\cup\undirected{F'}$.  By
    Claim~\ref{clm:F-and-F-bar}, it follows that
    $(p_1,q_1),(p_2,q_2)\in E\cup F'$ and
    $(p_1,q_2),(p_2,q_1)\notin E\cup F'$.  Since $F'\subset F$, we have
    $(p_1,q_1),(p_2,q_2)\in E\cup F$.  Furthermore, from
    $(p_1,q_2),(p_2,q_1)\in P\times Q$ and $F'= F\cap (P\times Q)$, we
    conclude that $(p_1,q_2),(p_2,q_1)\notin E\cup F$.  By construction of
    our reduction and since we only insert arcs, we have
    $(p_1,w),(p_2,w)\in E\cup F$.  Together with the coloring
    $\sigma(p_1)=\sigma(p_2)\ne\sigma(q_1)=\sigma(q_2)=\sigma(w)$, the
    latter arguments imply that $(\G+F,\sigma)$ contains an induced
    F3-graph.  By Lemma~\ref{lem:forbidden-subgraphs}, this contradicts
    that $(\G+F,\sigma)$ is a BMG.  \hfill$\diamond$
  \end{claim-proof}

  Now, let $(F,\undirected{F})$ be a completion pair such that
  $|\undirected{F}|\le k$ and $\undirected{F}$ is a minimum-sized edge
  completion set for $U$. Thus
  $U'=(P\cupdot Q,\undirected{E}\cup\undirected{F})$ is a chain graph.
  Hence, Claim~\ref{clm:chain-implies-bmg} implies that $(\G+F,\sigma)$ is
  a BMG.  Since $|F|=|\undirected{F}|\le k$, it follows that
  \PROBLEM{$2$-BMG Completion} with input $(\G,\sigma,k)$ has a yes-answer
  if \PROBLEM{CGC} with input $(U=(P\cupdot Q,\undirected{E}),k)$ has a
  yes-answer.

  Finally, let $F$ be a minimum-sized arc completion set for $(\G,\sigma)$,
  i.e.\ $(\G+F,\sigma)$ is a BMG, and assume $|F|\le k$. This and
  Claim~\ref{clm:minimum-arc-compl} implies $F\subseteq P\times Q$.  For
  the corresponding completion pair $(F,\undirected{F})$ we have
  $|\undirected{F}|=|F|\le k$. Moreover, since $(\G+F,\sigma)$ is a BMG,
  Claim \ref{clm:bmg-implies-chain} implies that
  $U=(P\cupdot Q,\undirected{E}\cup \undirected{F})$ is a chain graph.
  Therefore, \PROBLEM{CGC} with input $(U=(P\cupdot Q,\undirected{E}),k)$
  has a yes-answer if \PROBLEM{$2$-BMG Completion} with input
  $(\G,\sigma,k)$ has a yes-answer. This completes the proof.
\end{proof}

\begin{figure}[t]
  \begin{center}
    \includegraphics[width=0.85\linewidth]{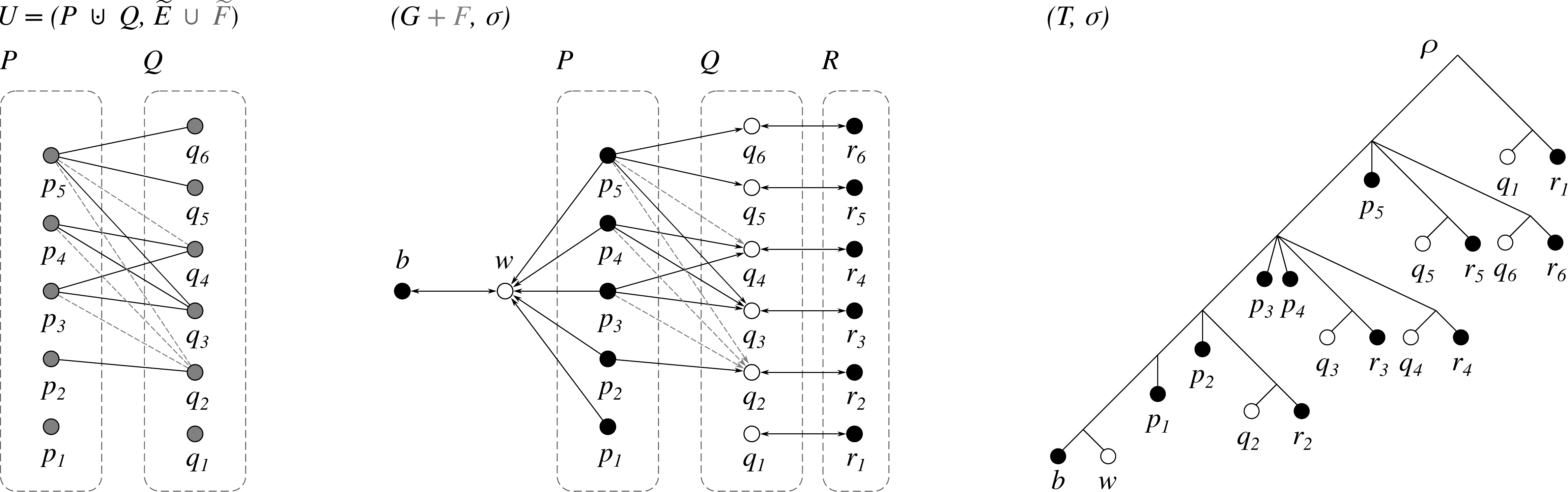}
  \end{center}
  \caption[]{An example solution for \PROBLEM{CGC}, resp., \PROBLEM{$2$-BMG
      Completion} as constructed in the proof of
    Thm.~\ref{thm:2col-completion-NP}.  A tree $(T,\sigma)$ that explains
    the resulting BMG is shown on the right.  Here, we have $k=4$ edge,
    resp., arc additions (indicated by dashed-gray lines) to obtain a
    chain-graph, resp., 2-BMG.  The indices of the vertices in
    $P=\{p_1,\dots,p_{|P|}\}$ are chosen w.r.t.\ the order $\lessdot$ on
    $P$ i.e.\ $i<j$ if and only if $p_i\lessdot p_j$ and thus,
    $N(p_i)\subseteq N(p_j)$.  In this example, we have
    $N(p_1)\cap Q=\emptyset$.  Moreover, the vertex $q_1$ has no neighbor
    in $P$.}
  \label{fig:completion-tree}
\end{figure}

  Not all (2-)BMGs can be explained by binary trees
  \cite{Schaller:2020z,Schaller:2020}.  BMGs $(\G,\sigma)$ for which a
  binary explaining tree exists have been termed \emph{binary-explainable}
  \cite{Schaller:2020z}. They are of practical significance because
  phylogenetic trees are often assumed to be binary by nature, with
  multifurcations arising in many cases as an artifact of insufficient data
  \cite{DeSalle:94,Hoelzer:94,Maddison:89}.

  We therefore consider the modified completion problem that, 
  given an arbitrary properly colored digraph $(\G,\sigma)$, aims to find a 
  binary-explainable BMG:
  
  \begin{problem}[\PROBLEM{$\ell$-BMG Completion restricted to 
  Binary-Explainable Graphs ($\ell$-BMG CBEG)}]\ \\
    \begin{tabular}{ll}
      \emph{Input:}    & A properly $\ell$-colored digraph $(\G =(V,E),\sigma)$
      and an integer $k$.\\
      \emph{Question:} & Is there a subset $F\subseteq V\times V \setminus
      (\{(v,v)\mid v\in V\} \cup E)$ such
      that $|F|\leq k$\\
      & and $(\G+ F,\sigma)$ is a binary-explainable $\ell$-BMG?
    \end{tabular}
  \end{problem}
  \noindent The corresponding editing and deletion problems will be
    called \PROBLEM{$\ell$-BMG EBEG} and \PROBLEM{$\ell$-BMG DBEG},
    respectively.

  Binary-explainable BMGs can be characterized in terms of a simple
  forbidden subgraph.
  \begin{definition}
    An \emph{hourglass} in a properly vertex-colored graph $(\G,\sigma)$,
    denoted by $[xy \hourglass x'y']$, is a subgraph $(\G[Q],\sigma_{|Q})$
    induced by a set of four pairwise distinct vertices
    $Q=\{x, x', y, y'\}\subseteq V(\G)$ such that (i)
    $\sigma(x)=\sigma(x')\ne\sigma(y)=\sigma(y')$, (ii) $(x,y),(y,x)$ and 
    $(x'y'),(y',x')$ are bidirectional arcs
    in $\G$, (iii) $(x,y'),(y,x')\in E(\G)$, and (iv)
    $(y',x),(x',y)\notin E(\G)$.
  \end{definition}
  \noindent
  Fig.~\ref{fig:hourglass} illustrates this definition.
  A graph $(\G,\sigma)$ is called \emph{hourglass-free} if it does not contain 
  an hourglass as an induced subgraph.
  \begin{proposition}{\cite[Prop.~8]{Schaller:2020}}
    \label{prop:binary-iff-hourglass-free}
    A BMG $(\G,\sigma)$ can be explained by a binary tree if and only if it is 
    hourglass-free.
  \end{proposition}
  
  \begin{figure}[t]
    \begin{center}
      \includegraphics[width=0.35\linewidth]{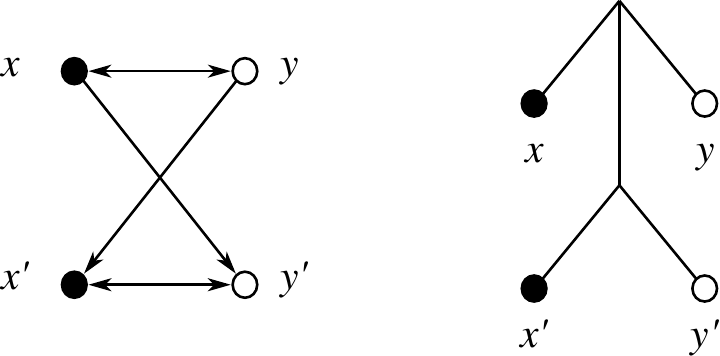}
    \end{center}
    \caption[]{An hourglass (left) and the unique non-binary tree (right) 
    that explains it.}
    \label{fig:hourglass}
  \end{figure}

  The reduction employed in the proof of Thm.~\ref{thm:2col-completion-NP} can 
  be adapted to show that the \PROBLEM{$2$-BMG CBEG} problem is hard.
  
  \begin{corollary}
    \label{cor:2col-completion-BEG-NP}
    \PROBLEM{$2$-BMG CBEG} is NP-complete.
  \end{corollary}
  \begin{proof}
    As shown in \cite[Cor.~6]{Schaller:2020} and
    \cite[Cor.~3.6]{Schaller:2020z}, binary-explainable BMGs can be
    recognized in polynomial time.  Therefore, \PROBLEM{$2$-BMG CBEG} is
    contained in the class NP.
    
    To show hardness of the problem, we use the same reduction from
    \PROBLEM{CGC} and the same arguments as in the proof of
    Thm.~\ref{thm:2col-completion-NP}.  In addition, we observe that the
    hourglass $[xy \hourglass x'y']$ contains the  bidirectional arcs
    $(x,y)$ and $(y,x)$ and each of the two vertices $x$ and $y$ has two
    out-neighbors, and thus, also at least two out-neighbors in every
    graph that contains the hourglass as an induced subgraph.
  
    We have to show that \PROBLEM{CGC} with input
    $(U=(P\cupdot Q,\undirected{E}),k)$ has a yes-answer if and only if
    \PROBLEM{$2$-BMG CBEG} with input $(\G,\sigma,k)$ as constructed in the
    proof of Thm.~\ref{thm:2col-completion-NP} has a yes-answer.  Recall
    that by Claims~\ref{clm:bmg-implies-chain}
    and~\ref{clm:chain-implies-bmg}, $(\G+F,\sigma)$ is a BMG if and only
    if $U'=(P\cupdot Q,\undirected{E}\cup\undirected{F})$ is a chain graph,
    where $(F,\undirected{F})$ is a completion pair.  Moreover, by
    Claim~\ref{clm:minimum-arc-compl}, every minimum-sized arc completion
    set $F$ for which $(\G+F,\sigma)$ is a BMG satisfies
    $F\subseteq P\times Q$.  Therefore, we can again argue via minimal
    completion pairs $(F,\undirected{F})$ to conclude that, both in the
    \emph{if}- and in the \emph{only-if}-direction, we have a $2$-BMG
    $(\G+F,\sigma)$ with $F\subseteq P\times Q$, i.e., we only inserted
    arcs from $P$ to $Q$.  Using Fig.~\ref{fig:completion-tree}, it is now
    easy to verify that every bidirectional pair of arcs in $(\G+F,\sigma)$
    is either incident to the vertex $b$ or to one of the vertices in $R$.
    Moreover, every vertex in $R\cup\{b\}$ has exactly one out-neighbor.
    The latter two arguments together with the observation that hourglasses
    require bidirectional arcs $(x,y),(y,x)$ such that both $x$ and $y$
    have at least two out-neighbors imply that $(\G+F,\sigma)$ must be
    hourglass-free.  Therefore, $(\G+F,\sigma)$ is binary-explainable by
    Prop.~\ref{prop:binary-iff-hourglass-free}, which completes the proof.
\end{proof}

	Results analogous to Cor.~\ref{cor:2col-completion-BEG-NP}
  cannot be derived as easily for \PROBLEM{$2$-BMG EBEG} and
  \PROBLEM{$2$-BMG DBEG}.  The reason is that the graph $(\G,\sigma)$
  (as well as $(G\symdiff F,\sigma)$) constructed in the proof of
  Thm.~\ref{thm:2col-NP} contains a large number of induced hourglasses. In
  particular, for every $i\in\{1,\dots,m\}$, we have that every black
  vertex $a\in X_i$, white vertex $b\in X_i$, black vertex $a'\in Y_i$ and
  white vertex $b'\in Y_i$ induce an hourglass $[ab \hourglass a'b']$. Even
  though we suspect these problems to be hard as well, a different
  reduction will be required to prove this conjecture.

\section{Complexity of $\ell$-BMG modification problems}
\label{sect:NPcomplete-all}

We now turn to the graph modification problems for an arbitrary number $\ell$
of colors.  The proof of the next theorem follows the same strategy of
adding hub-vertices as in \cite{Hellmuth:2020a}.
\begin{theorem}
  \label{thm:ell-col-NPc}
  \PROBLEM{$\ell$-BMG Deletion}, \PROBLEM{$\ell$-BMG Completion}, and
  \PROBLEM{$\ell$-BMG Editing} are
  NP-complete for all $\ell\geq 2$.
\end{theorem}
\begin{proof}
  BMGs can be recognized in polynomial time by 
  Cor.~\ref{cor:bmg-rec-polytime}
  and thus, all three problems are contained in the class
  NP.  Let $(\G=(V,E),\sigma)$ be a properly colored digraph with $\ell$
  colors.  Thm.~\ref{thm:2col-NP}, Cor.~\ref{cor:2col-NP-deletion} and
  Thm.~\ref{thm:2col-completion-NP} state NP-completeness for the case of
  $\ell=2$ colors.  Thus assume $\ell\ge3$ in the following.

  By slight abuse of notation, we collectively refer to the three
  problems \PROBLEM{$\ell$-BMG Deletion}, \PROBLEM{$\ell$-BMG
  Completion}, and \PROBLEM{$\ell$-BMG Editing} simply as
  \PROBLEM{$\ell$-BMG Modification}.  Correspondingly, we write
  $(\G\odot F, \sigma)$ and distinguish the three problems by the
  modification operation $\odot\in\{- , +, \symdiff \}$, where $\odot=-$,
  $\odot=+$ and $\odot=\symdiff$ specifies that $F$ is a deletion-,
  completion, or edit set, respectively.

  We use reduction from \PROBLEM{$2$-BMG Modification}. To this end, let
  $(\G_2=(V_2,E_2),\sigma_2,k)$ be an instance of one of the latter three
  problems.  To obtain a properly colored graph
  $(\G_{\ell}=(V_{\ell},E_{\ell}),\sigma_{\ell})$ with $\ell$ colors, we
  add to $\G_2$ a set $V_H$ of $\ell-2$ new vertices with pairwise
  distinct colors that also do not share any colors with the vertices in
  $(\G_2,\sigma_2)$.  Moreover, we add arcs such that every $h\in V_H$
  becomes a hub-vertex.  Note that $V_{\ell}=V_2\cupdot V_H$,
  $\G_{\ell}[V_2]=\G_2$, and $(\sigma_{\ell})_{|V_2}=\sigma_2$.
  Furthermore, $V_2$ is a subset of $V_{\ell}$ satisfying the condition in
  Obs.~\ref{obs:color-restriction}, i.e.,
  $V_2= \bigcup_{s\in S_2} V_{\ell}[s]$ for the color set $S_2$ in
  $(\G_2,\sigma_2)$. Clearly, the reduction can be performed in
  polynomial time.  We proceed by showing that an instance
  $(\G_2,\sigma_2,k)$ of the respective \PROBLEM{$2$-BMG Modification}
  problem has a yes-answer if and only if the corresponding instance
  $(\G_{\ell},\sigma_{\ell},k)$ of \PROBLEM{$\ell$-BMG Modification} has a
  yes-answer.

  Suppose that \PROBLEM{$2$-BMG Modification} with input
  $(\G_2,\sigma_2,k)$ has a yes-answer. Then there is an arc set
  $F\subseteq V_2\times V_2 \setminus \{(v,v)\mid v\in V_2\}$ with
  $|F|\le k$ such that $(\G_2\odot F, \sigma_2)$ is a BMG.  Let
  $(T_2,\sigma_2)$ be a tree with root $\rho$ explaining
  $(\G_2\odot F, \sigma_2)$.  Now take $(T_2,\sigma_2)$ and add the
    vertices in $V_H$ as leaves of the root $\rho$ and color these leaves
    as in $(\G_{\ell},\sigma_{\ell})$, to obtain the tree
    $(T_{\ell},\sigma_{\ell})$. 
  By construction, we have $L(T_{\ell})=V_{\ell}=V_2\cup V_H$ and
  $T_2=(T_{\ell})_{|V_2}$, where $(T_{\ell})_{|V_2}$ is the restriction of
  $T_{\ell}$ to the leaf set $V_2$.  The latter arguments together with
  Obs.~\ref{obs:color-restriction} imply that
  $(\G(T_{\ell},\sigma_{\ell})[V_2],(\sigma_{\ell})_{|V_2})=
  \G((T_{\ell})_{|V_2},(\sigma_{\ell})_{|V_2}) =\G(T_2,\sigma_2)=(\G_2\odot
  F, \sigma_2)$.

  Let $h\in V_H$ be arbitrary.  Since $h$ is the only vertex of its color,
  $(x,h)$ is an arc in $\G(T_{\ell},\sigma_{\ell})$ for every
  $x\in V_{\ell}\setminus \{h\}$.  Since $h$ is a child of the root, we
  have moreover $\lca_{T_{\ell}}(x,h)=\rho$, and thus, $(h,x)$ is an arc in
  $\G(T_{\ell},\sigma_{\ell})$ for every $x\in V_{\ell}\setminus \{h\}$.
  The latter two arguments imply that $h$ is a hub-vertex in
  $\G(T_{\ell},\sigma_{\ell})$.  Since $F$ is not incident to any vertex in
  $V_{\ell}\setminus V_2=V_H$ and each vertex $h\in V_H$ is a hub-vertex in
  $(\G_{\ell},\sigma_{\ell})$ and in $\G(T_{\ell},\sigma_{\ell})$, we
  conclude that
  $\G(T_{\ell},\sigma_{\ell})=(\G_{\ell}\odot F,\sigma_{\ell})$.  Hence,
  $(\G_{\ell}\odot F,\sigma_{\ell})$ is a BMG and the corresponding
  \PROBLEM{$\ell$-BMG Modification} problem with input
  $(\G_{\ell},\sigma_{\ell},k)$ has a yes-answer.

  For the converse, suppose that \PROBLEM{$\ell$-BMG Modification}
  with input $(\G_{\ell},\sigma_{\ell},k)$ has a
  yes-answer.  Thus, there is an arc set
  $F\subseteq V_{\ell}\times V_{\ell} \setminus \{(v,v)\mid v\in
  V_{\ell}\}$ with $|F|\le k$ such that $(\G_{\ell}\odot F, \sigma_{\ell})$
  is a BMG.  Let $(T_{\ell},\sigma_{\ell})$ be a tree explaining
  $(\G_{\ell}\odot F, \sigma_{\ell})$.  Let $F'\subseteq F$ be the subset
  of arc modifications $(x,y)$ for which $x,y\in V_2$.  Thus, it holds
  $|F'|\le|F|\le k$. By construction,
  $(\G_{\ell}\odot F)[V_2]=\G_{\ell}[V_2]\odot F'$.  Moreover, by
  Obs.~\ref{obs:color-restriction}, we have
  $(\G(T_{\ell},\sigma_{\ell})[V_2],(\sigma_{\ell})_{|V_2})=
  \G((T_{\ell})_{|V_2},(\sigma_{\ell})_{|V_2})$.  In summary, we obtain
  $(\G_2\odot F',\sigma_2)=(\G_{\ell}[V_2]\odot F',\sigma_2)=
  ((\G_{\ell}\odot F)[V_2],(\sigma_{\ell})_{|V_2}) =
  (\G(T_{\ell},\sigma_{\ell})[V_2],(\sigma_{\ell})_{|V_2}) =
  \G((T_{\ell})_{|V_2},(\sigma_{\ell})_{|V_2})$.  Thus,
  $(\G_2\odot F',\sigma_2)$ is a BMG.  Together with $|F'|\le k$, this
  implies that \PROBLEM{$2$-BMG Modification} with input
	$(\G_2,\sigma_2,k)$ has a yes-answer.
\end{proof}

As in the $2$-colored case, we can reuse the reduction to show that $\ell$-BMG 
CBEG is NP-complete.
\begin{corollary}
  \label{cor:ell-col-CBEG-NPc}
  \PROBLEM{$\ell$-BMG CBEG} is NP-complete for all $\ell\geq 2$.
\end{corollary}
\begin{proof}
  As shown in \cite[Cor.~6]{Schaller:2020} and
  \cite[Cor.~3.6]{Schaller:2020z}, binary-explainable BMGs can be
  recognized in polynomial time.  Therefore, \PROBLEM{$\ell$-BMG CBEG} is
  contained in the class NP. Cor.~\ref{cor:2col-completion-BEG-NP} states
  NP-completeness for the case $\ell=2$.  Thus, it remains to show
  NP-hardness for the case $\ell\ge 3$.
  
  We use a reduction from \PROBLEM{$2$-BMG CBEG} and the same
  polynomial-time construction as in the proof of
  Thm.~\ref{thm:ell-col-NPc}, i.e., we construct an $\ell$-colored graph
  $(\G_{\ell}=(V_{\ell},E_{\ell}),\sigma_{\ell})$ from a $2$-colored graph
  $(\G_{2}=(V_{2},E_{2}),\sigma_{2})$ by adding a hub-vertex of $\ell-2$
  pairwise distinct new colors. Note that we only consider arc insertions
  here, i.e., we have $\odot=+$.  We proceed by showing that an instance
  $(\G_2,\sigma_2,k)$ of the respective \PROBLEM{$2$-BMG CBEG} problem has
  a yes-answer if and only if the corresponding instance
  $(\G_{\ell},\sigma_{\ell},k)$ of \PROBLEM{$\ell$-BMG CBEG} has a
  yes-answer.
  
  First suppose that \PROBLEM{$2$-BMG CBEG} with input $(\G_2,\sigma_2,k)$
  has a solution
  $F\subseteq V_2\times V_2 \setminus \{(v,v)\mid v\in V_2\}$ with
  $|F|\le k$ such that $(\G_2+ F, \sigma_2)$ is a \emph{binary-explainable}
  BMG. By Prop.~\ref{prop:binary-iff-hourglass-free}, $(\G_2+F, \sigma_2)$
  is hourglass-free.  Since $(\G_2+F, \sigma_2)$ is in particular a BMG, we
  can use the same arguments as in the proof of Thm.~\ref{thm:ell-col-NPc}
  to conclude that $(\G_{\ell}+F, \sigma_{\ell})$ is a BMG.  Now observe
  that an hourglass contains two vertices of each of its two colors.
  Therefore and since every vertex in $V_H=V_{\ell}\setminus V_{2}$ is the
  only vertex of its color, none the vertices in $V_H$ is part of an
  induced subgraph of $(\G_{\ell}+ F, \sigma_{\ell})$ that is an hourglass.
  Hence, all hourglasses of $(\G_{\ell}+F, \sigma_{\ell})$ must be part of
  the induced subgraph $((\G_{\ell}+F)[V_2],(\sigma_{\ell})_{|V_2})$.  This
  together with the facts that
  $((\G_{\ell}+F)[V_2],(\sigma_{\ell})_{|V_2})=(\G_2+F, \sigma_2)$ and
  $(\G_2+F, \sigma_2)$ is hourglass-free implies that
  $(\G_{\ell}+F, \sigma_{\ell})$ must also be hourglass-free.  By
  Prop.~\ref{prop:binary-iff-hourglass-free}, the BMG
  $(\G_{\ell}+F, \sigma_{\ell})$ is binary-explainable, and hence,
  \PROBLEM{$\ell$-BMG CBEG} with input $(\G_{\ell},\sigma_{\ell},k)$ has a
  yes-answer.

  For the converse, suppose that \PROBLEM{$\ell$-BMG CBEG} with input
  $(\G_{\ell},\sigma_{\ell},k)$ has a solution
  $F\subseteq V_{\ell}\times V_{\ell} \setminus \{(v,v)\mid v\in
  V_{\ell}\}$ with $|F|\le k$ such that $(\G_{\ell}+ F, \sigma_{\ell})$ is
  a binary-explainable BMG.  By Prop.~\ref{prop:binary-iff-hourglass-free},
  $(\G_{\ell}+ F, \sigma_{\ell})$ is hourglass-free.  As before, let
  $F'\subseteq F$ be the subset of arc modifications $(x,y)$ for which
  $x,y\in V_2$.  By the same arguments as in the proof of
  Thm.~\ref{thm:ell-col-NPc}, we have $|F'|\le |F|\le k$ and
  $(\G_2+F',\sigma_2)= ((\G_{\ell}+F)[V_2],(\sigma_{\ell})_{|V_2})$ is a
  BMG.  In particular, $(\G_2+F',\sigma_2)$ is an induced subgraph of
  $(\G_{\ell}+ F, \sigma_{\ell})$, and thus, hourglass-free.  Together with
  Prop.~\ref{prop:binary-iff-hourglass-free}, the latter arguments imply
  that \PROBLEM{$2$-BMG CBEG} with input $(\G_{2},\sigma_{2},k)$ has a
  yes-answer.
\end{proof}

It remains an open question whether an analogous result holds for the
corresponding editing and deletion problems restricted to
binary-explainable graphs.  However, by the same arguments as in the proof 
of Cor.~\ref{cor:ell-col-CBEG-NPc}, we have
\begin{remark}
  If \PROBLEM{$2$-BMG EBEG} is NP-complete, then 
  \PROBLEM{$\ell$-BMG EBEG} is NP-complete for all $\ell\geq 2$.
  If \PROBLEM{$2$-BMG DBEG} is NP-complete, then 
  \PROBLEM{$\ell$-BMG DBEG} is NP-complete for all $\ell\geq 2$.
\end{remark}

\section{ILP formulation of $\ell$-BMG modification problems}
\label{sect:ILP}

Hard graph editing problems can often be solved with integer linear
programming (ILP) on practically relevant instances. It is of interest,
therefore, to consider an ILP formulation of the BMG deletion, completion
and editing problems considered above.  As input, we are given an
$\ell$-colored digraph $(\G=(V,E),\sigma)$.  We encode its arcs by the
binary constants
\begin{align*}
E_{xy}=1 \text{ if and only if } (x,y)\in E.
\end{align*}
for all pairs $(x,y)\in V\times V$, $x\neq y$. The vertex coloring
$\sigma$ is represented by the binary constant
\begin{align*}
\varsigma_{y,s}=1  \text{ if and only if } \sigma(y)=s
\end{align*}
The arc set of the modified graph $(\G^*,\sigma)$ is encoded by binary
variables $\epsilon_{xy}$, that is, $\epsilon_{xy} = 1$ if and only if
$(x,y)$ is arc in the modified graph $\G^*$.  The aim is to minimize the
number of edit operations, and thus, the symmetric difference
between the respective arc sets. This is represented by the objective
function
\begin{align}
  &\min \sum_{(x,y)\in V\times V} (1-\epsilon_{xy})E_{xy} +
    \sum_{(x,y)\in V\times V} (1-E_{xy})\epsilon_{xy}.
\end{align}
The same objective function can also be used for the BMG completion
and BMG deletion problem.  To ensure that only arcs between vertices of
distinct colors exist, we add the constraints
\begin{align}
  \epsilon_{xy}=0 \text{ for all } (x,y)\in V\times V \text{ with }
  \sigma(x)=\sigma(y).\label{eq:proper-color}
\end{align}
For the BMG completion problem, the arc set $E$ must be contained in the
modified arc set. Hence, we add
\begin{align}
  E_{xy}\leq \epsilon_{xy} \text{ for all } (x,y)\in V\times V. \label{ilp:add}
\end{align}
In this case, Equ.~\eqref{ilp:add} ensures that $\epsilon_{xy}=1$
if $E_{xy}=1$ and thus, $(x,y)$ remains an arc in the modified graph.
In contrast, for the BMG deletion problem,
it is not allowed to add arcs and thus, we use
\begin{align}
  \epsilon_{xy} \leq E_{xy} \text{ for all } (x,y)\in V\times V. \label{ilp:del}
\end{align}
In this case, Equ.~\eqref{ilp:del} ensures that $\epsilon_{xy}=0$ if
$E_{xy}=0$ and thus, $(x,y)$ does not become an arc in the modified graph.
For the BMG editing problem, we neither need Constraint \eqref{ilp:add} nor
\eqref{ilp:del}.
Both characterization in Thm.~\ref{thm:newCharacterizatio} 
and~\ref{thm:BMG-charac-via-R-F}, respectively, require that
  $(\G^*,\sigma)$ is sf-colored.  Eq.~\ref{eq:proper-color} already ensures
  a proper coloring and thus, it remains to make sure that each vertex has
  at least one out-neighbor of every other color. This property translates
to the constraint
\begin{align}
  &\sum_{y\neq x} \epsilon_{xy}\cdot\varsigma_{y,s} 
  >0\label{ilp:all-colors}
\end{align}
for all $s\ne\sigma(x)$.

The $O(|V|^2)$ variables and $O(|V|^2)$ constraints introduced above are 
  relevant for
  $\ell$-BMG modification problems for an arbitrary $\ell$.  In the
  following two subsections, we present additional constraints and
  variables that are sufficient for the cases $\ell=2$ and $\ell\ge 2$,
  respectively.

\subsection*{$2$-BMG modification problems}

By Thm.~\ref{thm:newCharacterizatio}, a properly 2-colored graph is a
  BMG if and only if it is sink-free and does not contain an induced F1-,
  F2-, or F3-graph.  Equ.~\eqref{ilp:all-colors} already guarantees
that $(\G^*,\sigma)$ is sink-free.  Hence it suffices to add
constraints that exclude induced F1-, F2-, and F3-graphs.  For every
ordered four-tuple $(x_1, x_2, y_1, y_2) \in V^4$ with pairwise distinct
$x_1, x_2, y_1, y_2$ and
$\sigma(x_1)=\sigma(x_2)\neq \sigma(y_1)=\sigma(y_2)$, we require
\begin{align}
\text{(F1)\qquad}
& \epsilon_{x_1y_1}+\epsilon_{y_1x_2} + \epsilon_{y_2x_2} +
(1-\epsilon_{x_1y_2})+ (1-\epsilon_{y_2x_1})\le 4 \text{ and }\label{ilp:F1}\\
\text{(F2)\qquad}
& \epsilon_{x_1y_1}+ \epsilon_{y_1x_2} + \epsilon_{x_2y_2} +
(1-\epsilon_{x_1y_2})\le 3.\label{ilp:F2}
\intertext{In addition, for every ordered
  five-tuple $(x_1, x_2, y_1, y_2,y_3) \in V^5$ with pairwise distinct
  $x_1, x_2, y_1, y_2,y_3$ and
  $\sigma(x_1)=\sigma(x_2)\neq \sigma(y_1)=\sigma(y_2)=\sigma(y_3)$, we
  enforce}
\text{(F3)\qquad}
& \epsilon_{x_1y_1}+\epsilon_{x_1y_3} + \epsilon_{x_2y_2}+ \epsilon_{x_2y_3}
+ (1-\epsilon_{x_1y_2}) + (1-\epsilon_{x_2y_1})\le 5.\label{ilp:F3}
\end{align}

By construction, we still have $O(|V|^2)$ variables but $O(|V|^5)$
  constraints.  We note that the $2$-colored case is handled correctly by
  the ILP formulation for the general $\ell$-colored case given in the next
  section.  However, the additional variables required for $\ell>2$ are not
  needed here.  We note in passing that, accordingly, we observed a
significant speedup when compared to the application of the general
formulation to $2$-colored graphs in a cursory simulation.

\subsection*{General $\ell$-BMG modification problems}

For the general $\ell$-colored case, we drop 
Equations~(\ref{ilp:F1})-(\ref{ilp:F3}), and instead rely on
Thm.~\ref{thm:BMG-charac-via-R-F}, which requires that the pair
  $(\mathscr{R}(\G^*,\sigma), \mathscr{F}(\G^*,\sigma))$ is compatible.
To implement this constraint, we follow the approach of \cite{chang2011ilp}
and \cite{Hellmuth:15}. Note that we make no distinction between the two
triples $ba|c$ and $ab|c$.  In order to avoid superfluous variables and
symmetry conditions connecting them, we assume that the first two indices
in triple variables are ordered. Thus there are three triple variables
$t_{ab|c}$, $t_{ac|b}$ and $t_{bc|a}$ for any three distinct
$a, b, c \in V$.  We add constraints such that $t_{ab|c}=1$ if $ab|c$ is an
informative triple (cf.\ Def.~\ref{def:informative_triples} and Lemma
\ref{lem:informative_triples}) and $t_{ab|c}=0$ if $ab|c$ is a
  forbidden triple (cf.\ Def.~\ref{def:forbidden_triples} and
  Lemma~\ref{lem:forb_triples}).  Hence, we add
\begin{align}
  & \epsilon_{xy}+(1-\epsilon_{xy'}) - t_{xy|y'} \leq 1 \text{ and}
  \label{ilp:info-t}\\
  & \epsilon_{xy}+\epsilon_{xy'} + t_{xy|y'} \leq 2 \label{ilp:forb-t}
\end{align}
for all ordered $(x,y,y')\in V^3$ with three pairwise distinct vertices
$x,y,y'$ and $\sigma(x)\neq \sigma(y)=\sigma(y')$.  Equ.~\eqref{ilp:info-t}
ensures that if $(x,y)$ is an arc ($\epsilon_{xy}=1$) and $(x,y')$ is not
an arc ($\epsilon_{xy'}=0$) in the edited graph, then $t_{xy|y'}=1$.
To obtain a BMG, we must ensure that there is a tree that displays all
  triples in $\mathscr{R}(\G^*,\sigma)$ and none of the triples in
  $\mathscr{F}(\G^*,\sigma)$.

  A phylogenetic tree $T$ is uniquely determined by its sets of clusters
  $\mathscr{C}(T) \coloneqq \{L(T(v)) \mid v\in V(T)\}$ \cite{Semple:03}.
  Thus, it is possible to reconstruct $T$ by building the clusters induced
  by the informative triples while avoiding that forbidden triples are
  displayed. The set of clusters $\mathscr{C}(T)$ forms a hierarchy, that
  is, for all $p,q\in \mathscr{C}(T)$ it holds that
  $p\cap q\in\{\emptyset, p,q\}$. It is easy to see that, in order to
  recover $T$ from $\mathscr{C}(T)$, it suffices to take into account only
  the non-trivial clusters $p\in \mathscr{C}(T)$ with $|p|\neq 1$ and
  $p\neq V(T)$.  The number of non-trivial clusters of $T$ is bounded by
  $L(T)-2$ (cf.\ \cite[Lemma~1]{Hellmuth:15}), where $L(T)=V$, i.e., the
  vertex set of the $\ell$-colored digraph $(\G=(V,E),\sigma)$ given as
  input.  In order to translate the condition that $\mathscr{C}(T)$ forms
  hierarchy into the language of ILPs, we follow
  \cite{chang2011ilp,Hellmuth:15}.  Let $M$ be a binary
  $|V| \times (|V|-2)$ matrix with entries $M(x,p) = 1$ iff vertex $x\in V$
  is contained in cluster $p$.  Each cluster $p$ of the tree $T_M$ encoded
  by $M$, which is represented by the $p$-th column of $M$, corresponds to
  an inner vertex $v_p$ in $T$ so that
  $L(T(v_p )) = \{x \mid x\in V, M(x,p)=1\}$.  In the following, we
  identify column $p$ with the corresponding cluster $L(T(v_p))$.

  We next ensure that all informative triples and none of the forbidden
  triples $ab|c$ are displayed by $T_M$. This is case if and only if there
  exists an inner vertex $v_p$ such that $a,b \in L(T(v_p))$ and
  $c\notin L(T(v_p))$ for every informative triple and no such vertex
  exists for any forbidden triple.  Therefore, we define, for all ordered
  three-tuples $(a,b,c) \in V^3$ and all $p\in\{1,\dots,|V|-2\}$, the
  binary variable $m((ab|c),p)$ and set $m((ab|c), p) = 1$ iff
  $M(a, p) = M (b, p) =1$ and $M (c, p) = 0$, i.e., iff the cluster $p$
  contains $a$ and $b$ but not $c$. The latter can be achieved by adding,
  for all these variables, the constraint
  \begin{align}
    & 0 \leq - 3 \cdot m((ab|c), p) + M (a, p) + M (b, p) + (1 - M (c, p)) \leq 
    2.
    \label{ilp:triples-clusters}
  \end{align}
  Full enumeration of all possible values that can be assigned to
  $M(a, p)$, $M (b, p)$ and $M(c, p)$ shows that $m((ab|c), p) =1$ if and
  only if $M(a, p) = M (b, p) =1$ and $M (c, p) = 0$.

  For every informative triple $ab|c$ there must be at least one column $p$
  for which $m((ab|c), p) = 1$ and for each forbidden triple it must be
  ensured that $m((ab|c), p) = 0$ for all $p\in\{1,\dots,|V|-2\}$. This is
  achieved by adding
  \begin{align}
    & t_{ab|c} \leq \sum_{p=1}^{|V|-2}  m((ab|c), p) \leq (|V|-2)\cdot t_{ab|c}
  \end{align} 
  for all ordered $(a,b,c)\in V^3$ with three pairwise distinct vertices
  $a,b,c$ and $\sigma(a)\neq \sigma(b)=\sigma(c)$.  If $t_{ab|c}=1$, then
  $m((ab|c), p)=1$ for at least one $p$ and if $t_{ab|c}=0$ then,
  $\sum_{p=1}^{|V|-2} m((ab|c), p) \leq 0$ implies that all $m((ab|c), p)$
  are put to $0$.

  Finally, we must ensure that the matrix $M$ indeed encodes the hierarchy
  of a tree.  This is the case if all clusters $p$ and $q$ are compatible,
  i.e., if $p\cap q\in\{p, q,\emptyset\}$.  Equivalently, two clusters $p$
  and $q$ are incompatible if there are vertices $a$, $b$ and $c$ such that
  $a\in p\setminus q$, $b\in q\setminus p$ and $c\in p\cap q$, which is
  represented by the ``gametes'' $(M(a,p),M(a,q))=(1,0)$,
  $(M(b,p),M(b,q))=(0,1)$ and $(M(c,p),M(c,q))=(1,1)$.  We avoid such
  incompatible clusters by using the so-called three-gamete condition which
  is described e.g.\ in \cite{Gusfield:97} or \cite[SI]{Hellmuth:15}.  To
  this end, we add for each of the three gametes
  $(\Gamma,\Lambda)\in\{(0,1),(1,0),(1,1)\}$ the binary variables
  $C(p,q,\Gamma\Lambda)$ for every pair of columns $p\ne q$.  Furthermore,
  we add the constraints
  \begin{align}
    & C(p,q,01) \ge - M(a,p) + M(a,q)\\
    & C(p,q,10) \ge M(a,p) - M(a,q)\\
    & C(p,q,11) \ge M(a,p) + M(a,q) - 1
  \end{align}
  for every pair of columns $p\ne q$ and every $a\in V$.  This ensures that
  $C(p,q,\Gamma\Lambda)=1$ whenever $M(a,p)=\Gamma$ and $M(a,q)=\Lambda$
  holds for at least one $a\in V$.  Finally, we add the constraint
  \begin{align}
    & C_{p,q,01} + C_{p,q,10} + C_{p,q,11} \le 2
  \end{align}
  for every pair of columns $p\ne q$, in order to ensure the compatibility
  of clusters $p$ and $q$.
  
  In total, this ILP formulation requires $O(|V|^4)$ variables and
  $O(|V|^4)$ constraints where the most expensive part stems from the
  variables $m((ab|c), p)$ and their corresponding constraints (cf.\
  Equ.~\ref{ilp:triples-clusters}).

\section{Concluding Remarks}

We have shown here that arc modification problems for BMGs are
NP-complete. This is not necessarily an obstacle for using BMG editing in
practical workflows -- after all, the computational problems in
phylogenetics all involve several NP-complete steps, including
\textsc{Multiple Sequence Alignment} \cite{Elias:06} and the
\textsc{Maximum Parsimony Tree} \cite{Graham:82} or \textsc{Maximum
  Likelihood Tree} problems \cite{Chor:06}. Nevertheless, highly efficient
and accurate heuristics have been devised for these problems, often
adjusted to the peculiarities of real-life data, so that the computational
phylogenetics have become a routine task in bioinformatics. As a starting
point to tackling BMG editing in practice, we gave novel characterizations
of BMGs that made it possible to introduce an ILP formulation that should
be workable at least for moderate-size instances.  Without further
  optimization, instances with $\ell\ge 3$ colors are very demanding and
  already problems with ten vertices may take a few hours on a desktop
  system. The $2$-colored version (requiring only $O(|V|^2)$ variables) on
  the other hand, handles instances with 20 vertices in about a minute. 
	We tested both versions using IBM ILOG
  CPLEX\texttrademark{ }Optimizer 12.10 and Gurobi Optimizer 9.0, and
  applied them to randomly disturbed ($2$-)BMGs.

We note in passing that \PROBLEM{2-BMG Deletion} and \PROBLEM{2-BMG
  Completion} can be shown to be fixed-parameter tractable (with the number
$k$ of edits as parameter) provided that the input graph is sink-free. To
see this, observe that sink-free 2-colored graphs are BMGs if and only if
they do not contain induced F1-, F2-, and F3-subgraphs (cf.\
Thm.~\ref{thm:newCharacterizatio}).  The FPT result follows directly from
the observation that all such subgraphs are of fixed size and only a fixed
number of arc deletions (resp., additions) are possible. In the case of
\PROBLEM{2-BMG Deletion}, only those arc deletions are allowed that do not
produce sinks in $G$. Clearly, graphs remain sink-free under arc addition.
It remains unclear whether \PROBLEM{2-BMG Editing} is also FPT for
sink-free graphs. One difficulty is that arc deletions may result in a
sink-vertex which then needs to be resolved by subsequent arc additions. It
also remains an open question for future research whether the BMG
modification problems for (not necessarily sink-free) $\ell$-colored graphs
are also FPT. We suspect that this is not the case for $\ell\ge3$, where
the characterization also requires consistency of the set of informative
triples. Since removal of a triple from $\mathscr{R}(\G,\sigma)$ requires
the insertion or deletion of an arc, it seems difficult to narrow down the
editing candidates to a constant-size set. Indeed, \textsc{Maximum Triple
  Inconsistency} is not FPT when parametrized by the number $k$ of triples
to be excluded \cite{Byrka:10}. On the other hand, the special case of
\textsc{Dense Maximum Triple Inconsistency} is FPT \cite{Guillemot:13}. The
set of informative triples $\mathscr{R}(\G,\sigma)$, however, is usually
far from being dense.

For larger-scale practical applications, we expect that heuristic
algorithms will need to be developed. An interesting starting point is the
observation that in many examples some of the (non-)arcs in forbidden
subgraphs cannot be modified. This phenomenon of unambiguously identifiable
(non-)arcs will be the topic of ongoing work.

\section*{Acknowledgments}

We thank Nicolas Wieseke for stimulating discussions.  This work was funded
in part by the German Research Foundation (DFG).

\clearpage
\bibliographystyle{plain}
\bibliography{preprint2-bmg-edit}

\begin{thebibliography}{10}

\bibitem{Abrams:10}
Gene Abrams and Jessica~K. Sklar.
\newblock The graph menagerie: Abstract algebra and the mad veterinarian.
\newblock {\em Math. Mag.}, 83:168--179, 2010.

\bibitem{Aho:81}
A.V. Aho, Y.~Sagiv, T.G. Szymanski, and J.D. Ullman.
\newblock Inferring a tree from lowest common ancestors with an application to
  the optimization of relational expressions.
\newblock {\em SIAM J Comput}, 10:405--421, 1981.

\bibitem{Altenhoff:16}
Adrian~M Altenhoff, Brigitte Boeckmann, Salvador Capella-Gutierrez, Daniel~A
  Dalquen, Todd DeLuca, Kristoffer Forslund, Jaime Huerta-Cepas, Benjamin
  Linard, C{\'e}cile Pereira, Leszek~P Pryszcz, Fabian Schreiber, Alan Sousa~da
  Silva, Damian Szklarczyk, Cl{\'e}ment-Marie Train, Peer Bork, Odile Lecompte,
  Christian von Mering, Ioannis Xenarios, Kimmen Sj{\"o}lander, Lars
  Juhl~Jensen, Maria~J Martin, Matthieu Muffato, {Quest for Orthologs
  consortium}, Toni Gabald{\'o}n, Suzanna~E Lewis, Paul~D Thomas, Erik
  Sonnhammer, and Christophe Dessimoz.
\newblock Standardized benchmarking in the quest for orthologs.
\newblock {\em Nature Methods}, 13:425--430, 2016.

\bibitem{Byrka:10}
Jaroslaw Byrka, Sylvain Guillemot, and Jesper Jansson.
\newblock New results on optimizing rooted triplets consistency.
\newblock {\em Discr. Appl. Math.}, 158:1136--1147, 2010.

\bibitem{chang2011ilp}
Wen-Chieh Chang, Gordon~J Burleigh, David~F Fern{\'a}ndez-Baca, and Oliver
  Eulenstein.
\newblock An {ILP} solution for the gene duplication problem.
\newblock {\em BMC Bioinformatics}, 12(S1):S14, 2011.

\bibitem{Chor:06}
Benny Chor and Tamir Tuller.
\newblock Finding a maximum likelihood tree is hard.
\newblock {\em J. ACM}, 53:722--744, 2006.

\bibitem{Cohn:02}
Henry Cohn, Robin Pemantle, and James~G. Propp.
\newblock Generating a random sink-free orientation in quadratic time.
\newblock {\em Electr. J. Comb.}, 9:R10, 2002.

\bibitem{DeSalle:94}
R.~DeSalle, R.~Absher, and G.~Amato.
\newblock Speciation and phylogenetic resolution.
\newblock {\em Trends Ecol. Evol.}, 9:297--298, 1994.

\bibitem{dondi2017approximating}
Riccardo Dondi, Manuel Lafond, and Nadia El-Mabrouk.
\newblock Approximating the correction of weighted and unweighted orthology and
  paralogy relations.
\newblock {\em Alg. Mol. Biol.}, 12:4, 2017.

\bibitem{El-Mallah:1988}
E.~S. El-Mallah and C.~J. Colbourn.
\newblock The complexity of some edge deletion problems.
\newblock {\em IEEE Trans. Circuits Syst.}, 35:354--362, 1988.

\bibitem{Elias:06}
Isaac Elias.
\newblock Settling the intractability of multiple alignment.
\newblock {\em J. Comput. Biol.}, 13:1323--1339, 2006.

\bibitem{Fitch:70}
W~M Fitch.
\newblock Distinguishing homologous from analogous proteins.
\newblock {\em Syst Zool}, 19:99--113, 1970.

\bibitem{Geiss:2019a}
Manuela Gei{\ss}, Edgar Ch{\'a}vez, Marcos Gonz{\'a}lez~Laffitte, Alitzel
  L{\'o}pez~S{\'a}nchez, B{\"a}rbel M.~R. Stadler, Dulce~I. Valdivia, Marc
  Hellmuth, Maribel Hern{\'a}ndez~Rosales, and Peter~F. Stadler.
\newblock Best match graphs.
\newblock {\em J. Math. Biol.}, 78:2015--2057, 2019.

\bibitem{BMG-corrigendum}
Manuela Gei{\ss}, Edgar Ch{\'a}vez, Marcos Gonz{\'a}lez~Laffitte, Alitzel
  L{\'o}pez~S{\'a}nchez, B{\"a}rbel M.~R. Stadler, Dulce~I. Valdivia, Marc
  Hellmuth, Maribel Hern{\'a}ndez~Rosales, and Peter~F. Stadler.
\newblock Best match graphs {\textit{(corrigendum)}}.
\newblock arxiv.org/1803.10989v4, 2020.

\bibitem{Geiss:2020c}
Manuela Gei{\ss}, Marcos E.~Gonz{\'a}lez Laffitte, Alitzel~L{\'o}pez
  S{\'a}nchez, Dulce~I. Valdivia, Marc Hellmuth, Maribel~Hern{\'a}ndez Rosales,
  and Peter~F. Stadler.
\newblock Best match graphs and reconciliation of gene trees with species
  trees.
\newblock {\em J. Math. Biol.}, 80:1459--1495, 2020.

\bibitem{Geiss:2020b}
Manuela Gei{\ss}, Peter~F. Stadler, and Marc Hellmuth.
\newblock Reciprocal best match graphs.
\newblock {\em J. Math. Biol.}, 80:865--953, 2020.

\bibitem{Graham:82}
R.~L. Graham and L.~R. Foulds.
\newblock Unlikelihood that minimal phylogenies for a realistic biological
  study can be constructed in reasonable computational time.
\newblock {\em Math. Biosci.}, 60:133--142, 1982.

\bibitem{Guillemot:13}
Sylvain Guillemot and Matthias Mnich.
\newblock Kernel and fast algorithm for dense triplet inconsistency.
\newblock {\em Theor. Comp. Sci.}, 494:134--143, 2013.

\bibitem{Gusfield:97}
D~Gusfield.
\newblock {\em Algorithms on Strings, Trees and Sequences: Computer Science and
  Computational Biology}.
\newblock Cambridge University Press, Cambridge UK, 1997.

\bibitem{He:06}
Ying-Jun He, Trinh N~D Huynh, Jesper Jansson, and Wing-Kin Sung.
\newblock Inferring phylogenetic relationships avoiding forbidden rooted
  triplets.
\newblock {\em J. Bioinf. Comp. Biol.}, 4:59--74, 2006.

\bibitem{Hellmuth:2020a}
Marc Hellmuth, Manuela Gei{\ss}, and Peter~F. Stadler.
\newblock Complexity of modification problems for reciprocal best match graphs.
\newblock {\em Theor. Comp. Sci.}, 809:384--393, 2020.

\bibitem{Hellmuth:13a}
Marc Hellmuth, Maribel Hernandez-Rosales, Katharina~T. Huber, Vincent Moulton,
  Peter~F. Stadler, and Nicolas Wieseke.
\newblock Orthology relations, symbolic ultrametrics, and cographs.
\newblock {\em J. Math. Biol.}, 66:399--420, 2013.

\bibitem{Hellmuth:15}
Marc Hellmuth, Nicolas Wieseke, Marcus Lechner, Hans-Peter Lenhof, Martin
  Middendorf, and Peter~F. Stadler.
\newblock Phylogenomics with paralogs.
\newblock {\em Proc Natl Acad Sci USA}, 112:2058--2063, 2015.

\bibitem{Hoelzer:94}
G.~Hoelzer and D.~Meinick.
\newblock Patterns of speciation and limits to phylogenetic resolution.
\newblock {\em Trends Ecol. Evol.}, 9:104--107, 1994a.

\bibitem{Karp1972}
Richard~M. Karp.
\newblock Reducibility among combinatorial problems.
\newblock In Raymond~E. Miller, James~W. Thatcher, and Jean~D. Bohlinger,
  editors, {\em Complexity of Computer Computations: Proceedings of a symposium
  on the Complexity of Computer Computations}, pages 85--103. Springer, Boston,
  MA, 1972.

\bibitem{lafond2016link}
Manuel Lafond, Riccardo Dondi, and Nadia El-Mabrouk.
\newblock The link between orthology relations and gene trees: a correction
  perspective.
\newblock {\em Alg. Mol. Biol.}, 11:4, 2016.

\bibitem{lafond2014orthology}
Manuel Lafond and Nadia El-Mabrouk.
\newblock Orthology and paralogy constraints: satisfiability and consistency.
\newblock {\em BMC Genomics}, 15(S6):S12, 2014.

\bibitem{lafond2013gene}
Manuel Lafond, Magali Semeria, Krister~M Swenson, Eric Tannier, and Nadia
  El-Mabrouk.
\newblock Gene tree correction guided by orthology.
\newblock {\em BMC Bioinformatics}, 14(S15):S5, 2013.

\bibitem{Liu:2012}
Yunlong Liu, Jianxin Wang, Jiong Guo, and Jianer Chen.
\newblock Complexity and parameterized algorithms for {Cograph Editing}.
\newblock {\em Theor. Comp. Sci.}, 461:45--54, 2012.

\bibitem{Maddison:89}
W.~Maddison.
\newblock Reconstructing character evolution on polytomous cladograms.
\newblock {\em Cladistics}, 5:365--377, 1989.

\bibitem{Natanzon:2001}
Assaf Natanzon, Ron Shamir, and Roded Sharan.
\newblock Complexity classification of some edge modification problems.
\newblock {\em Discr. Appl. Math.}, 113:109--128, 2001.

\bibitem{Schaller:2020z}
David Schaller, Manuela Gei{\ss}, Marc Hellmuth, and Peter~F. Stadler.
\newblock Best {{Match Graphs}} with {{Binary Trees}}.
\newblock arXiv: 2011.00511, 2020.

\bibitem{Schaller:2020}
David Schaller, Manuela Gei{\ss}, Peter~F. Stadler, and Marc Hellmuth.
\newblock Complete characterization of incorrect orthology assignments in best
  match graphs.
\newblock {\em J. Math. Biol.}, 2020.
\newblock in press.

\bibitem{Semple:03}
Charles Semple and Mike Steel.
\newblock {\em Phylogenetics}.
\newblock Oxford University Press, Oxford UK, 2003.

\bibitem{Setubal:18a}
Jo{\~a}o~C. Setubal and Peter~F. Stadler.
\newblock Gene phylogenies and orthologous groups.
\newblock In Jo{\~a}o~C. Setubal, Peter~F. Stadler, and Jens Stoye, editors,
  {\em Comparative Genomics}, volume 1704, pages 1--28. Springer, Heidelberg,
  2018.

\bibitem{Sonnhammer:14}
Erik Sonnhammer, Toni Gabald{\'o}n, Alan Wilter Sousa~da Silva, Maria Martin,
  Marc Robinson-Rechavi, Brigitte Boeckmann, Paul Thomas, Christophe Dessimoz,
  and {Quest for Orthologs Consortium}.
\newblock Big data and other challenges in the quest for orthologs.
\newblock {\em Bioinformatics}, 30:2993--2998, 2014.

\bibitem{Stadler:2020}
Peter~F. Stadler, Manuela Gei{\ss}, David Schaller, Alitzel
  L{\'o}pez~S{\'a}nchez, Marcos Gonz{\'a}lez~Laffitte, Dulce~I. Valdivia, Marc
  Hellmuth, and Maribel Hern{\'a}ndez~Rosales.
\newblock From pairs of most similar sequences to phylogenetic best matches.
\newblock {\em Alg. Mol. Biol.}, 15:5, 2020.

\bibitem{Yannakakis:1981}
Mihalis Yannakakis.
\newblock Computing the {Minimum Fill-In} is {NP}-complete.
\newblock {\em SIAM J. Algebraic Discr. Methods}, 2:77--79, 1981.

\end{thebibliography}


\clearpage
\appendix

\section{All forbidden subgraphs in 2-colored BMGs}
\begin{figure}[ht]
	\begin{center}
  \includegraphics[width=0.97\linewidth, angle=90]{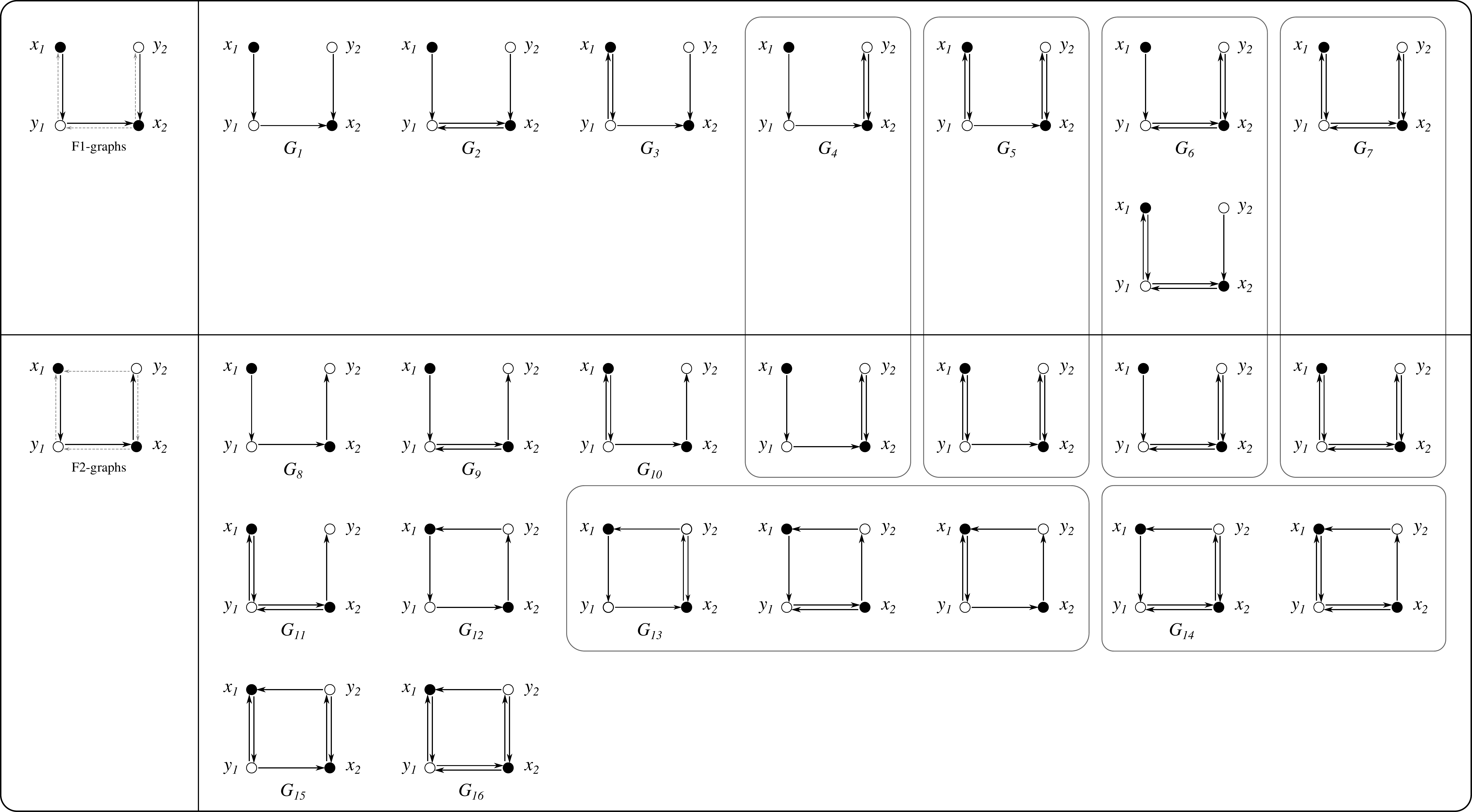}
	\end{center}
  \caption[]{All F1-graphs and F2-graphs. Isomorphism classes are indicated
    by the boxes, and labeled according to
    Fig.~\ref{fig:all-17}.}
  \label{fig:all-F1-F2}
\end{figure}

\begin{figure}[ht]
  \includegraphics[width=\linewidth]{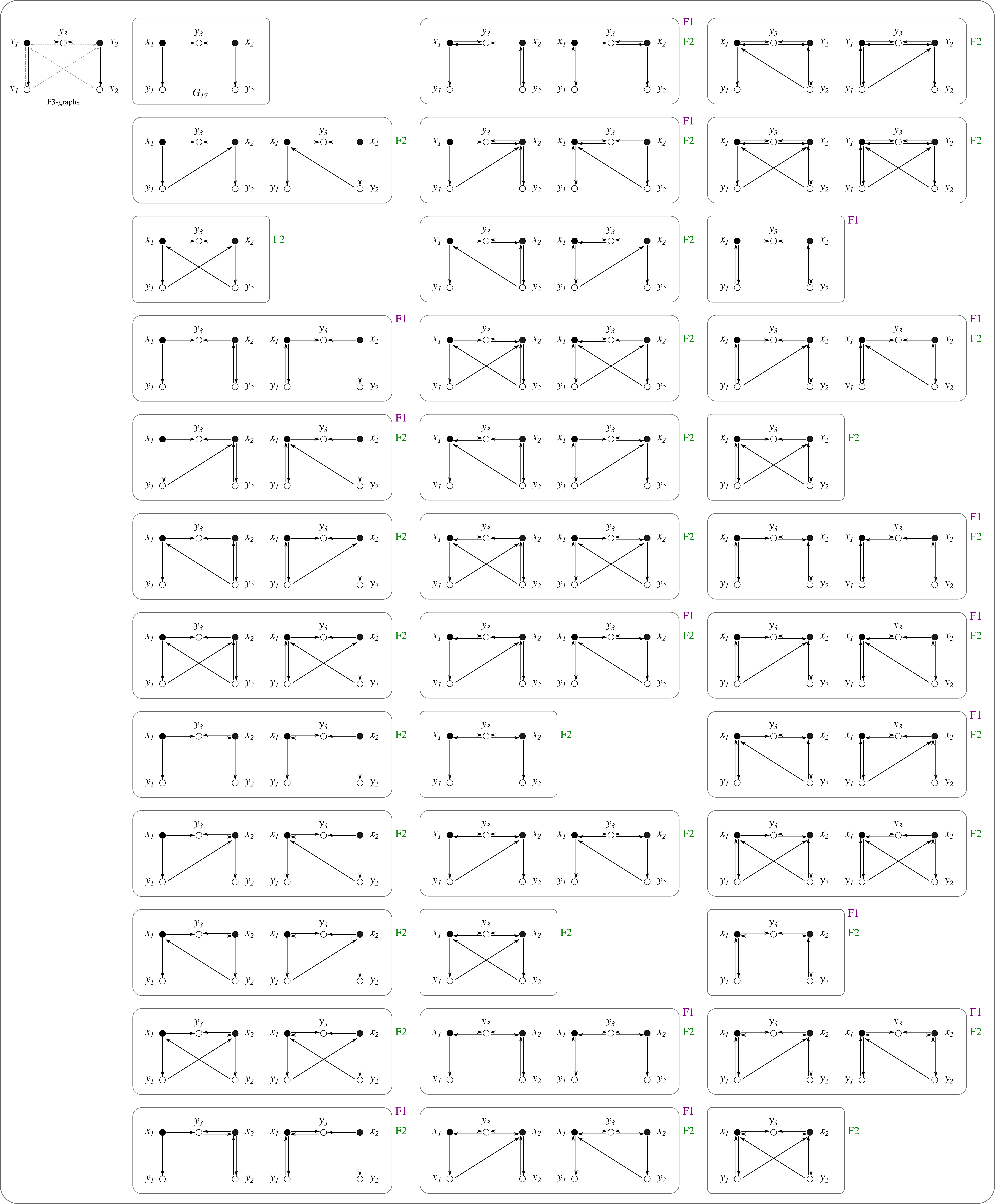}
  \caption[]{All F3-graphs. Isomorphism classes are indicated by the boxes.
    Those graphs that contain at least one F1- or F2-graph as an induced
    subgraph are marked with ``F1'', resp. ``F2''.}
  \label{fig:all-F3}
\end{figure}

\end{document}